\documentclass[11pt]{article}
 
\usepackage[a4paper,left=22mm,top=25mm,right=22mm,bottom=35mm, marginparsep=.1in]{geometry}

\usepackage{authblk}
\usepackage[numbers, sort&compress]{natbib}
\usepackage[margin=30pt,font=small,labelfont=bf]{caption}
\usepackage{amsmath,amsopn,amsthm,amssymb}

\usepackage{graphicx}
\usepackage{mathptmx} 
\usepackage[mathscr]{euscript}
\usepackage{mathtools}
\usepackage{float}
\usepackage{enumitem}
\setlist[enumerate,1]{label={(\arabic*)}}
\usepackage{xspace}
\usepackage{scalerel}
\usepackage{marginnote}
\usepackage{enumitem}

\usepackage[
 	colorlinks=true,
	urlcolor=black,
	linkcolor=blue
]{hyperref}

\usepackage{algorithmicx}
\usepackage{algorithm} 
\usepackage[noend]{algpseudocode}
\algrenewcommand\algorithmicrequire{\textbf{Input:}}
\algrenewcommand\algorithmicensure{\textbf{Output:}}

\makeatletter
\def\moverlay{\mathpalette\mov@rlay}
\def\mov@rlay#1#2{\leavevmode\vtop{%
    \baselineskip\z@skip \lineskiplimit-\maxdimen
    \ialign{\hfil$\m@th#1##$\hfil\cr#2\crcr}}}
\newcommand{\charfusion}[3][\mathord]{
  #1{\ifx#1\mathop\vphantom{#2}\fi
    \mathpalette\mov@rlay{#2\cr#3}
  }
  \ifx#1\mathop\expandafter\displaylimits\fi}
\makeatother

\newtheorem{theorem}{Theorem}[section]
\newtheorem{proposition}[theorem]{Proposition}
\newtheorem{lemma}[theorem]{Lemma}
\newtheorem{corollary}[theorem]{Corollary}
\newtheorem*{example*}{Example}
\newtheorem{definition}[theorem]{Definition}
\newtheorem{observation}[theorem]{Observation}

\newcommand{\reg}{\textsc{Reg}}
\newcommand{\cupdot}{\charfusion[\mathbin]{\cup}{\cdot}}
\newcommand{\lca}{\ensuremath{\operatorname{lca}}}

\newcommand{\LCA}{\ensuremath{\operatorname{LCA}}}

\newcommand{\Hasse}[1][]{\mathscr{H}\ifthenelse{\equal{#1}{}}{}{(#1)}}

\DeclareMathOperator{\child}{child}
\DeclareMathOperator{\parent}{par}
\DeclareMathOperator{\indeg}{indeg}
\DeclareMathOperator{\outdeg}{outdeg}
\DeclareMathOperator{\CC}{\mathtt{C}}

\newcommand{\lcaV}[1]{\mathfrak{lca}_{#1}}
\newcommand{\notlcaV}[1]{\overline{\mathfrak{lca}_{#1}}}
\newcommand{\vis}{\mathfrak{vis}}
\newcommand{\dvis}{\mathfrak{dvis}}
\newcommand{\notvis}{\overline{\mathfrak{vis}}}
\newcommand{\Xci}[1][i]{\mathcal{P}_{\leq #1}(X)}

\newcommand{\norm}{\textsc{Norm}}
\newcommand{\cov}{\textsc{Cov}}
\newcommand{\visop}{\ensuremath{\operatorname{vis}}}
\newcommand{\wvis}{\ensuremath{\operatorname{d-vis}}}

 \newenvironment{owndesc}%
    {\begin{description}[leftmargin = 0.2cm, labelsep = 0.2cm]}
    {\end{description}}

\providecommand{\keywords}[1]{\textbf{\textit{Keywords: }} #1}

\title{Regularizing and Normalizing DAGs and Phylogenetic Networks}

\author[1]{Marc Hellmuth} 

\author[2,*]{Anna Lindeberg} 

\author[3]{Vincent Moulton} 

\affil[1]{Faculty of Mathematics and Computer Science,
  		  Leipzig University, DE-04109 Leipzig, Germany}

\affil[2]{Department of Mathematics, Faculty of Science,
  Stockholm University, SE-10691 Stockholm, Sweden} 

\affil[3]{School of Computing Sciences, 
University of East Anglia, Norwich, NR4 7TJ, United Kingdom}

\affil[*]{corresponding author}

\date{\ }

\begin{document}
\sloppy

\maketitle

\abstract{ 
Phylogenetic networks and, more generally, directed acyclic graphs (DAGs)
represent hierarchical structure beyond trees, for instance in the presence
of reticulate evolutionary events such as hybridization or horizontal gene
transfer. A central question is which parts of such graphs are essential
with respect to leaf-observable information, and which parts can be removed
without changing this information. Resolving this question can lead to
principled simplification methods for phylogenetic networks, such as the
recent normalization approach of Francis et al.

In this paper, we study this question from three related perspectives:
clusters displayed by a DAG $G$, least common ancestors (LCAs) of subsets
of its leaf set, and visibility, a path-based property of vertices. We first
introduce an LCA-based simplification procedure called
\emph{$i$-regularization}. For a DAG $G$ and $i\geq 1$, the DAG
$\reg_i(G)$ retains precisely those vertices that occur as unique LCAs of
leaf subsets of size at most $i$, removes the remaining non-leaf vertices
by a graph-editing operation $\ominus$, and then deletes shortcuts. We show
that $\reg_i(G)$ 
admits a Hasse-diagram characterization in terms of the corresponding lca-clusters.

We then compare LCA-based regularization with normalization. Using the same
$\ominus$-operator, we describe the cover construction underlying
normalization, identify visible vertices that are nevertheless removed, and
characterize when regularization and normalization coincide. Together, these
results provide a unified framework for cluster-based, LCA-based, and
visibility-based simplifications of DAGs and phylogenetic networks.

\smallskip
\noindent
\keywords{Phylogenetic network, Regular network, Normal network, Cluster, Least common ancestors}


\section{Introduction}

Directed acyclic graphs (DAGs) and rooted phylogenetic networks are widely used to model
hierarchical structure in situations where purely tree-like representations are too restrictive. A
classical example arises in evolutionary biology, where reticulation events such as hybridization or
horizontal gene transfer cannot be captured by trees alone \cite{Huson:11}. Beyond phylogenetics,
DAG-based models also appear in areas such as database theory \cite{Aho:81}, workflow analysis and
process modeling \cite{aalst1998application,van2003workflow}, or in the theory of Bayesian and
causal networks \cite{pearl2009causality,koller2009probabilistic}.

A recurring theme in the study of such graphs is to determine which aspects of the underlying
topology can be recovered from ``coarse'' information derived from the leaves, and how one can
systematically simplify a given DAG or network without destroying the leaf-implied structure of
interest. In phylogenetics, the leaves represent present-day species (or other taxa), while internal
vertices encode ancestral ones. It is therefore natural to ask which internal vertices are
structurally essential and which are redundant with respect to a specified type of leaf-based
information.

Three classical viewpoints on this question are central to this paper. The first is \emph{cluster
information}: every vertex $v$ in a DAG $G$ with leaf set $X$ induces a cluster $\CC_G(v) = \{x\in X
\mid\, \text{there is a directed path from } v \text{ to } x \text{ in }G\}$, and the resulting set
system $\mathfrak{C}(G)$ records how leaves are grouped by ancestral vertices. The second is
\emph{least common ancestors (LCAs)}: for a subset $A\subseteq X$, its least common ancestors (when
they exist) identify the ``lowest'' vertices that jointly explain the ancestry of $A$. The third is
\emph{visibility}, a path-based notion that captures whether a vertex lies on all root-to-leaf paths
leading to at least one leaf. For example, in Figure~\ref{fig:normal-prune-intro2}, taking $G$ to be
the DAG $N_1$, the vertex $u$ induces the cluster $\{a,b\}$, the LCA of the vertices $a$ and $c$ is
$\rho$, and the vertex $u$ is visible, since all paths from the root $\rho$ to the leaf $a$ pass
through $u$.

\begin{figure}
    \centering
    \includegraphics[width=0.75\linewidth]{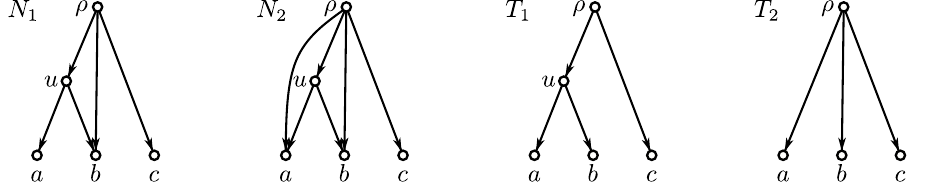}
    \caption{Two networks $N_1$, $N_2$ and two trees $T_1$, $T_2$, for which $T_1 =\reg(N_1) =
             \norm(N_1)$ and $T_1=\reg(N_2)\neq \norm(N_2)=T_2$, see Section~\ref{sec:reg}.}
    \label{fig:normal-prune-intro2}
\end{figure}

This leads to a fundamental question: given a DAG or network together with its induced leaf-based
information -- such as clusters, LCAs, or visibility -- which structural components are essential
for representing this information, and which are superfluous? In other words, can certain internal
vertices be removed or suitably transformed to yield a simpler directed graph that still faithfully
captures the prescribed leaf-based structure? From this perspective, simplifying a DAG or network
means defining a principled transformation that eliminates ``irrelevant'' substructures -- those
that do not contribute to the chosen leaf-based information -- while preserving the structural
relationships among the leaves. The main challenge is to formalize this notion of irrelevance and to
design simplification procedures that are canonical, well-behaved, and compatible with the
underlying leaf-based representation.

In~\cite{francis2021normalising}, Francis, Huson and Steel introduced a \emph{normalization}
procedure that transforms a network $N$ into a normal network $\norm(N)$ by retaining only visible
vertices, removing shortcuts, and suppressing vertices of in-degree and out-degree equal to one.
More recently, Lindeberg et al.~\cite{HL:24,HL:26} proposed an alternative simplification approach
based on retaining only those vertices that serve as least common ancestors for specified subsets of
leaves. A third, conceptually related, construction is to consider the Hasse diagram of the
clustering system $\mathfrak{C}(G)$, which yields a canonical DAG that is minimal in the sense of
representing all leaf-induced clusters. Note that there are other approaches for simplifying DAGs
and networks. Some focus on producing a tree from a given DAG or network, see for example
\cite{Gusfield:2014,Huber2019,DRESS2010535,Heiss:25,willson2011restrictedtrees,Baroni:06}, and
further simplifications that do not necessarily produce a tree are presented in \cite{Moret:04,
Willson2012CSDhomo,Willson:22}.

Although the normalization, the LCA-driven and the Hasse diagram approaches produce a simplified
version of the underlying DAG or network -- essentially by removing vertices deemed ``irrelevant''
and rewiring certain arcs -- each of them relies on different criteria. Normalization is based on
visibility and thus on path-based properties, whereas the LCA-driven simplification retains vertices
according to their role in explaining leaf subsets. The relationship between a DAG $G$ and the Hasse
diagram $\Hasse(\mathfrak{C}(G))$ of its cluster system has also been partially explored in
\cite{HL:24,HL:26,Baroni:05}; in essence, $\Hasse(\mathfrak{C}(G))$ retains single vertices $v$ for
which $\CC_G(v)$ is contained in $\mathfrak{C}(G)$ and its arcs reflect how such clusters are
related with respect to set inclusion. Our main contribution is to highlight the similarities
and differences between these three approaches to simplifying DAGs and networks, and to place them
within a unified conceptual framework.

In the following, we begin with the LCA-driven perspective and specialize the simplification
framework originally introduced in \cite{HL:24} to a systematic ``regularization'' approach
based on the notion of \emph{lca-relevance}. More specifically, for each integer $i \geq 1$, we
consider the set of vertices that arise as the unique LCA, denoted $\lca_G(A)$, for some subset $A
\subseteq X$ with $|A| \le i$, and we delete all other non-leaf vertices using the 
graph-editing operator $\ominus$ (which generalizes suppression and certain arc
contractions). After subsequently removing shortcuts (denoted with a superscripted ``$-$''), this
yields the \emph{$i$-regularization} $\reg_i(G) =(G\ominus W)^-$ of $G$. Intuitively,
$\reg_i(G)$ is a canonical simplification of $G$ that preserves the LCA structure of all leaf
subsets of size at most $i$.  Moreover, $\reg_i(G)$ is \emph{regular}, which
means that it is completely determined by its set of clusters.

Although regularity of $\reg_i(G)$ is an immediate consequence of results in \cite{HL:24},
that paper does not describe the set of clusters in
$\mathfrak{C}(\reg_i(G))$. To complete the picture, we give a cluster-level characterization of
$i$-regularization in Theorem~\ref{thm:clusters-of-prune}. Specifically, we show that $\reg_i(G)$ is
isomorphic to the Hasse diagram of certain \emph{lca-clusters}: clusters $C$ for which $\lca_G(C)$
is well-defined and forced by some subset $A$ with $|A|\le i$. Thus,
Theorem~\ref{thm:clusters-of-prune} provides a direct link between LCA-based simplification and the
classical cluster/Hasse-diagram viewpoint. This is complemented by
Theorem~\ref{thm:reg_vs_hasse}, which characterizes when the Hasse diagram $\Hasse(\mathfrak{C}(G))$
coincides with the fully regularized graph $\reg(G)\coloneqq\reg_{|X|}(G)$.

Next, we consider normal networks and their behavior under regularization. In
Proposition~\ref{prop:strong-normal=>2lcarel}, we extend the connection between normal and regular
networks established in Theorem~3.9 of \cite{Willson2010} by clarifying the relationships between
normality, strong normality, lca-relevance, and regularity. We then show in
Theorem~\ref{thm:normal=>Lrsf(N)-network} precisely how regularization acts on a normal network.
Besides providing further structural insight into normal networks, these results provide key
ingredients to characterize when normalization and regularization yield the same network, cf.\ 
Theorem~\ref{thm:norm(N)-ominus-form}.

We then turn to the principal goal of this paper: comparing
regularization with the normalization $\norm(N)$ of $N$. These two approaches are distinct. For
example, for the networks $N_1$ and $N_2$ in Figure~\ref{fig:normal-prune-intro2}, the networks
$\reg(N_1)$ and $\norm(N_1)$ coincide, whereas the networks $\reg(N_2)$ and $\norm(N_2)$ do not.
Although both approaches aim
to remove ``irrelevant'' vertices, they differ precisely in the criterion used to determine
which vertices are retained: normalization is based on visibility, a path-based notion that
captures whether a vertex lies on all root-to-leaf paths to some leaf, whereas $i$-regularization is
driven by lca-relevance.

To better understand the interplay between these two notions of relevance, in
Section~\ref{sec:normreg} we express both constructions within the common framework of
transformations of the form $(N\ominus W)^-$. The key difference lies in the
choice of the vertex set $W$. For $i$-regularization, $W$ consists, by definition, of the
vertices that are not unique LCAs of leaf subsets of size at most $i$. For normalization, we
first show that the cover construction $\cov(N)$ can be expressed as $
\cov(N)=(N\ominus\notvis(N))^-$, where $\notvis(N)$ is the set of non-visible vertices of $N$. We
then identify a subclass of visible vertices that are suppressed in the second step of
normalization, the so-called \emph{dispensably visible} vertices, and provide criteria for detecting
them directly in $N$. This yields an $\ominus$-based characterization of normalization itself and,
ultimately, allows us to characterize precisely when $\reg_i(N)$ and $\norm(N)$ coincide; see
Theorem~\ref{thm:norm(N)-ominus-form}, which is the main comparison result of the paper.

The rest of this paper is organized as follows. In Section~\ref{sec:basic} we recall the required
background on DAGs, shortcuts, normal networks, the $\ominus$-operator, clusters, and LCAs.
Section~\ref{sec:reg} introduces $i$-regularization, discusses its basic properties, and characterizes
$\reg_i(G)$ via lca-clusters and Hasse diagrams. It also discusses normal and strongly normal
networks and relates their structure to lca-relevance. In particular, we describe how normal
networks behave under $i$-regularization, that is, under the transformation $N\mapsto \reg_i(N)$. In
Section~\ref{sec:normreg} we then develop the framework for the comparison of network regularization
and normalization, derive $\ominus$-based formulas for $\cov(N)$ and $\norm(N)$, and identify
conditions under which regularization and normalization agree. Finally, in Section~\ref{sec:future}
we conclude with a discussion of some future directions.


\section{Preliminaries}
\label{sec:basic}

In this section, we recall the main definitions and results concerning DAGs, LCAs, clusters, and the
$\ominus$-operator that are needed throughout the paper. Most of this material is taken from, or
follows directly from, \cite{HL:24} and related work; it is included here for reference and to fix
notation. Those readers who are more familiar with these concepts can move on to the next section
and refer back to the relevant statements if necessary.

\paragraph{Basics.}
In what follows, $X$ will always be a finite non-empty set and $\Xci$ the set of non-empty
subsets of $X$ of size at most $i$. Moreover, $A\cupdot B$ denotes the \emph{disjoint union} of two
sets $A$ and $B$. A collection $\mathfrak{T}$ of non-empty subsets of $X$ is a \emph{hierarchy} if
$X\in\mathfrak{T}$, $\{x\}\in\mathfrak{T}$ for all $x\in X$, and for all $A,B\in \mathfrak{T}$, we
have $A\cap B\in \{\emptyset,A,B\}$.

A \emph{directed graph $G=(V,E)$} is an ordered tuple with non-empty vertex set $V(G)\coloneqq V$
and arc set $E(G)\coloneqq E\subseteq V\times V$. We sometimes use $u\to v$ to refer to the arc
$(u,v)\in E(G)$ and $u\leadsto v$ to denote a directed $uv$-path in $G$. We put
$\outdeg_G(v)\coloneqq\left|\left\{u\in V \colon (v,u)\in E\right\}\right|$ and
$\indeg_G(v)\coloneqq\left|\left\{u\in V \colon (u,v)\in E\right\}\right|$ to denote the
\emph{out-degree} and \emph{in-degree} of a vertex $v$, respectively. A vertex $v$ with
$\outdeg_G(v)=0$ is a \emph{leaf} of $G$ and a vertex $v$ with $\indeg_G(v)=0$ is a \emph{root} of
$G$. A \emph{hybrid} is a vertex $v$ with $\indeg_G(v)\geq 2$. Note that leaves might also be
hybrids. A directed graph $G$ is \emph{phylogenetic} if it does not contain a vertex $v$ such that
$\outdeg_G(v)=1$ and $\indeg_G(v)\leq 1$.

Directed graphs $G$ without directed cycles are called \emph{directed acyclic graphs (DAGs)}. If $G$
is a DAG with leaf set $X$, then $G$ is a \emph{DAG on $X$}. A \emph{network} is a DAG with a unique
root. A \emph{(rooted) tree} is a network that does not contain any hybrid vertices. A DAG $G$ is
\emph{separated} if all hybrids $v$ in $G$ have $\outdeg_G(v)=1$ \cite{Hellmuth2023}. 

Let $G$ be a DAG. We write $v\preceq_G w$ if and only if there is a directed $wv$-path in $G$. If
$v\preceq_G w$ and $v\neq w$, we write $v\prec_G w$. If $u\to v$ is an arc in $G$, then $v\prec_G u$
and we call $v$ a \emph{child} of $u$ and $u$ a \emph{parent} of $v$. In a DAG $G$ where $u\preceq_G
v$, we call $u$ a \emph{descendant} of $v$ and $v$ an \emph{ancestor} of $u$. If $u\preceq_G v$ or
$v\preceq_G u$, then $u$ and $v$ are \emph{$\preceq_G$-comparable} and, otherwise,
\emph{$\preceq_G$-incomparable}.

Two DAGs $G$ and $H$ are \emph{isomorphic} if there is an isomorphism between $G$ and $H$, i.e., a
bijective map $\varphi\colon V(G)\to V(H)$ such that $(u,v)\in E(G)$ if and only if
$(\varphi(u),\varphi(v))\in E(H)$. If this is the case, we write $G\simeq H$.
        
An arc $e=(u,w)$ in a DAG $G$ is a \emph{shortcut} if there is a directed $uw$-path that does not
contain the arc $e$ \cite{linz2020caterpillars, DOCKER2019129}. A DAG without shortcuts is
\emph{shortcut-free}. If $F\subseteq E$ is the set of all shortcuts of $G=(V,E)$, then we put
$G^-\coloneqq (V,E\setminus F)$.

\begin{lemma}[{\cite[L.~2.5]{HL:24} and \cite[L.~7.2]{willson2016comparing}}]\label{lem:properties-SF-G}\label{lem:same-prec-so-G=H}
	Let $G$ be a DAG on $X$. Then, $G^-$ is a shortcut-free DAG on $X$ such that, for all $u, v \in
	V(G)$, we have $u \preceq_{G} v$ if and only if $u \preceq_{G^-} v$. Moreover, if $G$ and $H$ are
	shortcut-free DAGs such that $V\coloneqq V(G)=V(H)$ and, for all $u,v\in V$, it holds that
	$u\preceq_G v$ if and only if $u\preceq_H v$, then $G=H$.
\end{lemma}

\paragraph{Normal Networks.}
Normal networks form a prominent class of phylogenetic networks
\cite{francis2021normalising,Willson2010}; see \cite{Francis:25} for a comprehensive overview.
Although our primary focus lies on normal \emph{networks}, several technical results are more
naturally formulated in the broader setting of DAGs. For this reason, we develop the theory at the
level of DAGs and specialize to networks where appropriate.

\begin{definition}
	A child $u$ of $v$ in a DAG $G$ is a \emph{tree-child} if $\indeg_G(u) = 1$. A DAG $G$ on $X$ is
	\emph{tree-child} if every vertex $v \in V(G) \setminus X$ has a tree-child \cite{Cardona:09}. A
	DAG is \emph{normal} if it is tree-child and shortcut-free \cite{Willson2010}.
\end{definition}

A common restriction in the phylogenetic literature, in particular for normal networks, is that
hybrid vertices are required to have a unique child. We explicitly call such networks
\emph{separated}. To avoid some technical issues caused by separated vertices, we introduce the
following class of DAGs:

\begin{definition}
	A DAG is \emph{strongly normal} if it is normal and does not contain any vertices that have only
	one child.
\end{definition}

Let $N$ be a network on $X$ with unique root $\rho$. A vertex $v$ of $N$ is \emph{visible (in $N$)}
if there exists a leaf $x\in X$ such that every $\rho x$-path of $N$ passes through $v$. Let
$\vis(N)$ denote the set of visible vertices of $N$ and put $\notvis(N)\coloneqq V(N)\setminus
\vis(N)$. Observe that, by definition, $\rho$ and every $x\in X$ are visible vertices of $N$. The
following result provides an interesting characterization of normal networks.

\begin{theorem}[{\cite[Lem.~2]{Cardona:09}}]\label{thm:char-normal-vis}
	A network is normal if and only if it is shortcut-free and all its vertices are visible. 
\end{theorem}

\paragraph{The $\ominus$-operator.}
We continue by summarizing important concepts from \cite{HL:24} where, in particular, the following
$\ominus$-operator is studied.

\begin{definition}[{\cite{SCHS:24}}]\label{def:ominus}
	Let $G=(V,E)$ be a DAG and $v\in V$. Then $G\ominus v=(V',E')$ is the directed graph with vertex
	set $V'=V\setminus\{v\}$ and arcs $(p,q)\in E'$ precisely if $v\ne p$, $v\ne q$ and $(p,q)\in E$,
	or if $(p,v)\in E$ and $(v,q)\in E$. For a non-empty subset $W = \{w_1,\dots,w_\ell\} \subsetneq
	V$, define $G\ominus W \coloneqq (\dots ((G \ominus w_1) \ominus w_2) \dots)\ominus w_\ell$. For
	notational reasons, put $G\ominus\emptyset=G$.
\end{definition} 

In simple terms, the directed graph $G \ominus v$ is obtained from $G = (V,E)$ by deleting the
vertex $v$ together with all incident arcs, and then adding an arc $p \to q$ for every parent $p$ of
$v$ and every child $q$ of $v$. If $v$ is a leaf or a root, we simply remove $v$ and its incident
arcs. If $v$ has in-degree one and out-degree one in $G$, then $G\ominus v$ is well-known as the
directed graph obtained from $G$ by \emph{suppressing} $v$. More generally, if $v$ has a unique
child $u$, then $G \ominus v$ can be viewed as contracting the arc $v \to u$ to the single vertex
$u$. Similarly, if $v$ has a unique \emph{parent} $u$, then $G\ominus v$ coincides with contracting
the arc $u\to v$ to $u$. Hence, the $\ominus$-operator unifies and generalizes several common
graph-editing operations on DAGs or networks.

It is also worth noting that the definition of $G \ominus W$ is motivated by the fact that the
$\ominus$-operator is commutative. This fact is mentioned –– without proof –– in
\cite{SCHS:24,HL:24}. For completeness, we shall prove this in the appendix of the paper.

\begin{lemma}\label{ominus:commutative}
	For all DAG $G$ and distinct vertices $u,v\in V(G)$ it holds that $(G \ominus v) \ominus u = (G
	\ominus u) \ominus v$.
\end{lemma}
\begin{proof}
	The proof is straightforward but tedious and, therefore, outsourced to Appendix~\ref{sec:appx}.
\end{proof}
  
Hence, the order of the vertices in $W$ does not affect the result of $G \ominus W \coloneqq (\dots
((G \ominus w_1) \ominus w_2) \dots) \ominus w_\ell$. It is easily seen that $G\ominus W$
remains a DAG and that the existence of directed paths in $G$ is preserved in $G\ominus W$. The
following lemma summarizes these simple but important properties of the $\ominus$-operator.
  
\begin{lemma}[{\cite[Obs~5.3]{HL:24}}]\label{lem:ominus-basics}
	Let $G$ be a DAG on $X$ and $W \subseteq V (G) \setminus X$. Then, $G \ominus W$ is a DAG on $X$
	that satisfies $u\preceq_{G} v\iff u\preceq_{G\ominus W} v$ for all $u,v \in V (G \ominus
	W)=V(G)\setminus W$. 
\end{lemma}

\paragraph{Regular Networks, LCAs and Clusters.}}

Let $G$ be a DAG on $X$. For every $v\in V(G)$ we define the \emph{cluster} of $v$ as the set
$\CC_G(v)\coloneqq\{x\in X\mid x\preceq_G v\}$. The set $\mathfrak{C}(G)=\{\CC_G(v)\mid v\in V(G)\}$
comprises all clusters of the vertices in $G$. 

Just as every DAG $G$ on $X$ naturally gives rise to a set system on $X$, we can also associate to
every set system on $X$ a corresponding DAG on $X$, as we now explain. The \emph{Hasse diagram}
$\Hasse(\mathfrak{C})$ of a set system $\mathfrak{C}$ on $X$ is the DAG with vertex set
$\mathfrak{C}$ and arc set containing an arc $A\to B$ for $A,B\in\mathfrak{C}$ if and only if (i)
$B\subsetneq A$ and (ii) there is no $C\in\mathfrak{C}$ with $B\subsetneq C\subsetneq A$. The Hasse
diagram $\Hasse(\mathfrak{C})$ is also known as the \emph{cover digraph} of $\mathfrak{C}$
\cite{Baroni:05}. Intuitively, DAGs that are indistinguishable from the Hasse diagram of their set
of clusters have much of their structure determined by these clusters. Such DAGs are called
\emph{regular}, and are formally defined as follows.

\begin{definition}[{\cite{Baroni:05}}]
  \label{def:regular-N}
  A DAG $G=(V,E)$ is \emph{regular} if the map
  $\varphi\colon V\to V(\Hasse[\mathfrak{C}(G)])$ defined by  $v\mapsto \CC_G(v)$ is an
  isomorphism between $G$ and $\Hasse(\mathfrak{C}(G))$.  
\end{definition}

Note that regular DAGs have been shown to be the only DAGs that can arise from certain simple models
of evolution, see \cite{Baroni:05,Baroni:06B} for details. As we shall see, regular DAGs are
also closely related to properties involving so-called least common ancestors, which we define
next. For a given DAG $G$ on $X$ and a non-empty subset $A\subseteq X$, a vertex $v\in V(G)$ is a
\emph{common ancestor of $A$} if $A\subseteq\CC_G(v)$, i.e., if $v$ is an ancestor of every vertex in
$A$. Moreover, $v$ is a \emph{least common ancestor} (LCA) of $A$ if $v$ is $\preceq_G$-minimal
among all common ancestors of $A$. The set $\LCA_G(A)$ comprises all LCAs of $A$ in $G$. 

\begin{figure}[ht]
  \centering
    \includegraphics[width=0.8\linewidth]{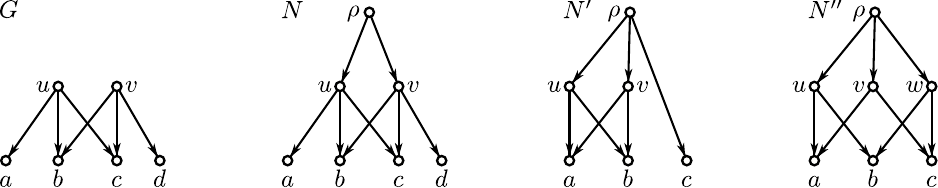}
    \caption{Four DAGs that illustrate certain concepts related to least common ancestors; see
             the text for further details.}
  \label{fig:exmpl-LCA}
\end{figure}

In general, not every set $A\subseteq X$ has a (unique) least common ancestor in a DAG; for example,
consider the DAG $G$ in Figure~\ref{fig:exmpl-LCA} where $|\LCA_G(\{b,c\})|>1$ and
$\LCA_G(\{a,d\})=\emptyset$. In a network $N$ on $X$, the unique root is a common ancestor for all
$A\subseteq X$ and, therefore, $\LCA_N(A)\neq\emptyset$. For simplicity, we will write $\lca_G(A)=v$
in case that $\LCA_G(A)=\{v\}$, in which case we say that \emph{$\lca_G(A)$ is well-defined};
otherwise, we leave $\lca_G(A)$ \emph{undefined}. Moreover, when $\lca_G(\{x,y\})$ is well-defined
for some $x,y\in X$, we write $\lca_G(x,y)$ to refer to $\lca_G(\{x,y\})$. A vertex $v\in V(G)$ is a
\emph{$k$-lca vertex} if $v=\lca_G(A)$ for some $A\subseteq X$ with $|A|=k$. In addition, a DAG $G$
satisfies the \emph{cluster-lca (CL)} property if $\lca_G(\CC_G(v))$ is well-defined for all $v\in
V(G)$.

Not all DAGs satisfy the (CL) property. For example, in the network $N'$ of
Figure~\ref{fig:exmpl-LCA}, $\LCA_{N'}(\CC_{N'}(v))=\LCA_{N'}(\{a,b\})=\{u,v\}$ and thus,
$\lca_{N'}(\CC_{N'}(v))$ is not well-defined. Nevertheless, as the following result shows,
$\lca_G(\CC_G(v))$ is well-defined whenever $v=\lca_G(A)$ for some subset $A\subseteq X$. This, and
the other properties stated in the following lemma, are simple consequences of the definitions and
the fact that $\CC_G(u)\subseteq\CC_G(v)$ for all $u,v\in V(G)$ such that $u\preceq_G v$.

\begin{lemma}[{\cite[Obs.~2]{SCHS:24} \& \cite[Cor.~3.6]{HL:24}}]\label{lem:lca-clusters}
	If $G$ is a DAG on $X$ and $A\subseteq X$ is a non-empty subset such that $\lca_G(A)$ is
	well-defined, then $\CC_G(\lca_G(A))$ is the unique inclusion-minimal cluster in $\mathfrak{C}(G)$
	containing $A$, and $\lca_G(\CC_G(\lca_G(A)))=\lca_G(A)$. In particular, $v\in V(G)$ is a $k$-lca
	vertex in $G$ for some $k$ if and only if $v = \lca_G (\CC_G (v))$. Moreover, if $u\in V(G)$ is a
	vertex such that $A\subseteq\CC_G(u)$, then $\lca_G(A)\preceq_G u$.
\end{lemma}

Note that, if $w$ is the unique child of a vertex $v$ in a DAG $G$, then $\CC_G(v)=\CC_G(w)$ holds.
This together with Lemma~3.4 and 3.5 in \cite{HL:24} implies the following natural property of
vertices with out-degree one.
\begin{lemma}\label{lem:outdeg1-no-lca}
	If $v$ is a vertex of a DAG $G$ on $X$ such that $\outdeg_G(v)=1$, then $v\notin\LCA_G(A)$ for any
	$A\subseteq X$. 
\end{lemma}

The following definition from \cite{HL:24} captures the notion of a DAG where every vertex is a
$k$-lca vertex for some $k\leq i$.

\begin{definition}\label{def:lca-rel}\label{def:2-lca-rel}
	A DAG $G=(V,E)$ on $X$ is $i$-lca-relevant if, for all $v\in V$, there is some $A\in \Xci$
	such that $v = \lca_G(A)$. 
\end{definition}

In other words, $G$ is $i$-lca-relevant precisely if every vertex of $G$ is a $k$-lca vertex for
some $k\leq i$. Hence, if $G$ is $i$-lca-relevant for some integer $i$, then $G$ is also
$j$-lca-relevant for all $j\geq i$. As shown by the DAG $N''$ in Figure~\ref{fig:exmpl-LCA}, the
converse of the latter statement is not always satisfied: $N''$ is 3-lca-relevant, but not
2-lca-relevant since its root $\rho$ is not the unique LCA for any subset consisting of one or two
leaves only. Note also that the DAG $G$ in Figure~\ref{fig:exmpl-LCA} is 2-lca-relevant, as
$\lca_G(x)=x$ for all $x\in\{a,b,c,d\}$, and $u=\lca_G(a,b)$, $v=\lca_G(c,d)$ are all well-defined.
This holds even though $\lca_G(b,c)$ is undefined. We call $|X|$-lca-relevant DAGs on $X$ simply
\emph{lca-relevant} since, in this case, each vertex $v$ is the unique LCA of just \emph{some}
subset of its leaves, with no restriction on the size of that set. By way of example, the DAGs $G$,
$N$ and $N''$ of Figure~\ref{fig:exmpl-LCA} are lca-relevant, while the network $N'$ is not.
Clearly, all $i$-lca-relevant DAGs are, in particular, lca-relevant.

We also note that there is a surprisingly close connection between regular and lca-relevant DAGs:

\begin{theorem}[{\cite[Thm.~4.10]{HL:24}}]\label{thm:regular-char}
    A DAG $G$ is regular if and only if $G$ is lca-relevant and shortcut-free.
\end{theorem}

We close this section by collecting together a couple of results that connect the $\ominus$-operator
with clusters and LCAs and that we shall need later on. The first observation is a direct
consequence of Lemma~\ref{lem:properties-SF-G} and Lemma~\ref{lem:ominus-basics}.

\begin{observation}\label{obs:prune-basics}
	Let $G= (V,E)$ be a DAG on $X$ and $W\subseteq V\setminus X$. Then, $H=(G\ominus W)^-$ is a
	shortcut-free DAG on $X$ that satisfies $u\preceq_G v \iff u\preceq_H v$ for all $u,v\in
	V(H)=V\setminus W$. Consequently, $\CC_H(v)=\CC_G(v)$ for all $v\in V\setminus W$ and, thus,
	$\mathfrak{C}(H)=\{\CC_G(v)\mid v\in V\setminus W \}\subseteq\mathfrak{C}(G)$. 
\end{observation}

In case $W$ solely contains vertices of out-degree one, we have further control over properties of
$G\ominus W$.
\begin{lemma}\label{lem:ominus-outdeg1}
	Let $G=(V,E)$ be a DAG on $X$ and let $W\subseteq\{v\in V\mid \outdeg_G(v)=1\}$. Then $G\ominus W$
	is a DAG on $X$ such that the following conditions hold.
    \begin{enumerate}
        \item $\LCA_{G\ominus W}(A)=\LCA_G(A)$ for all non-empty $A\subseteq X$.
        \item $\outdeg_{G\ominus W}(v)=\outdeg_G(v)$ for all $v\in V(G\ominus W)$.
    \end{enumerate}
\end{lemma}
\begin{proof}
	Let $G$ and $W$ be as stated. We first observe that by definition, $W$ cannot contain any leaves
	of $G$. Thus, Lemma~\ref{lem:ominus-basics} implies that $G\ominus W$ is a DAG on $X$. Now for the
	first statement, note that if $w\in W$, then $w\notin\LCA_G(Y)$ for any $Y\subseteq X$, see
	Lemma~\ref{lem:outdeg1-no-lca}. This together with \cite[Thm~5.5]{HL:24} ensures that
	$\LCA_{G\ominus W}(A)=\LCA_G(A)$ for all non-empty $A\subseteq X$.

	For Statement~(2), observe that removing a vertex $w$ of out-degree one with unique child $c$
	replaces, for every parent $p$ of $w$, the arc $p\to w$ by $p\to c$ and hence does not change the
	out-degree of $p$; the out-degrees of all other remaining vertices are unaffected. The claim
	therefore follows by repeatedly applying this observation to the vertices of $W$.
\end{proof}

In the final definition for this section, we provide a slight generalization of regular DAGs.

\begin{definition}[\cite{Hellmuth2023}]
	A DAG is \emph{semi-regular} if it is shortcut-free and satisfies the following condition 
  \begin{description}
		\item[(PCC):] for all $u, v \in V(G)$ it holds that $u$ and $v$ are $\preceq_G$-comparable if
		              and only if $\CC_G(u) \subseteq \CC_G(v)$ or $\CC_G(v) \subseteq \CC_G(u)$.
    \end{description}
\end{definition}

Theorem~4.6 of \cite{HL:24} shows that semi-regular DAGs generalize regular DAGs, i.e., that every
regular DAG is also semi-regular. As noted earlier, the DAGs $G$, $N$, and $N''$ in
Figure~\ref{fig:exmpl-LCA} are regular, and thus also semi-regular. In contrast, the network $N'$
is not semi-regular: it contains the $\preceq_{N'}$-incomparable vertices $u$ and $v$ that satisfy
$\CC_{N'}(u)=\{a,b\}=\CC_{N'}(v)$. Therefore, $N'$ violates the (PCC) property.

\section{Regularizations}
\label{sec:reg}

In this section, we consider a systematic way to transform a DAG into a regular DAG while
preserving its least common ancestor (LCA) structure up to some prescribed level. The regularization
procedure considered here is a specialization of the more general LCA-based simplification framework
introduced in \cite{HL:24}. We formulate this construction explicitly in terms of subsets of bounded
cardinality, which leads to the notion of $i$-regularization. Intuitively, this amounts to removing
vertices that do not play any role as LCAs of ``small'' subsets of leaves.

In Section~\ref{subs:regDAGs}, we recall the basic properties of this
construction and derive new connections between $i$-regularization, lca-clusters, and Hasse diagrams. In particular, we characterize
the set of clusters of $\reg_i(G)$ in terms of lca-clusters and show that $\reg_i(G)$ is isomorphic
to the Hasse diagram of this cluster system (see Theorem~\ref{thm:clusters-of-prune}). We further
characterize when the fully regularized DAG $\reg(G)$ coincides with the Hasse diagram of the
complete cluster system $\mathfrak{C}(G)$ (see Theorem~\ref{thm:reg_vs_hasse}).

In Section~\ref{sec:norm}, we investigate how LCA-based regularization acts on normal networks
and, in particular, clarify the relationship between normality, strong normality, and lca-relevance.
We show that if $N$ is a normal network, then $\reg(N)$ is again a normal network, and characterize
when $\reg(N)=N$ (see Theorem~\ref{thm:normal=>Lrsf(N)-network}). These structural results also
provide important ingredients for Section~\ref{sec:normreg}, where we compare regularization with
normalization within the common $\ominus$-framework.

\subsection{Regularizing DAGs}
\label{subs:regDAGs}

Recall that $\Xci$ denotes the non-empty subsets of $X$ of size at most $i$. For an integer
$i\geq 1$ and a DAG $G$, define the set
\[\lcaV{i}(G) \coloneqq \left\{v\in V(G)\mid v=\lca_G(A) \text{ for some } A\in \Xci\right\}
  \text{\ and \ } \notlcaV{i}(G) \coloneqq V(G)\setminus \lcaV{i}(G).\]
In other words, $\lcaV{i}(G)$ contains each vertex of $G$ that is a $k$-lca vertex for some $k\leq
i$ while $\notlcaV{i}(G)$ is the complement of that set. Note that $X\subseteq \lcaV{i}(G)$ for all
$i$, since every $x\in X$ satisfies $x=\lca_G(\{x\})$ for any DAG $G$ on $X$.

Another simple observation stems from Lemma~\ref{lem:outdeg1-no-lca}, namely, if a vertex $v$ has
out-degree one, then it is not an LCA of \emph{any} set of leaves. Consequently, vertices of
out-degree one are in particular always elements of $\notlcaV{i}(G)$ for any integer $i$ and any DAG
$G$. The next lemma gives a slight extension of this result which we will use later.

\begin{lemma}\label{lem:lca-sets-inclusions}
	Suppose $G$ is a DAG on $X$ and let $W=\{w\in V(G)\mid \outdeg_G(w)=1\}$. Then,
	\[\begin{aligned}
	X&=\lcaV{1}(G)\subseteq \lcaV{2}(G)\subseteq \ldots \subseteq \lcaV{|X|}(G)\subseteq V(G)\setminus W, &\text{and}\\ 
	W&\subseteq\notlcaV{|X|}(G)\subseteq\ldots\subseteq\notlcaV{2}(G)\subseteq\notlcaV{1}(G)=V(G)\setminus X& 
	\end{aligned}\]
\end{lemma}
\begin{proof}
	Let $G$ be a DAG on $X$, $|X|=\ell$ and $W=\{w\in V(G)\mid \outdeg_G(w)=1\}$. Note that
	$\lcaV{1}(G)=X$ follows from the fact that $v=\lca_G(\{x\})$ for some $x\in X$ if and only if
	$v=x$. Moreover, the inclusions $\lcaV{i}(G)\subseteq\lcaV{i+1}(G)$ for each $1\leq i< \ell$
	follow from the definition: if $v\in\lcaV{i}(G)$, then $v$ is a $j$-lca vertex for some $j$ with
	$1\leq j\leq i$ and thus clearly also for some $j$ with $1\leq j\leq i+1$, i.e.
	$v\in\lcaV{i+1}(G)$. Moreover, Lemma~\ref{lem:outdeg1-no-lca} implies that $w\notin\LCA_G(A)$ for
	all $w\in W$ and non-empty $A\subseteq X$. Consequently, $\lcaV{\ell}(G)\cap W=\emptyset$.
	Therefore, \[X=\lcaV{1}(G)\subseteq \lcaV{2}(G)\subseteq \ldots \subseteq \lcaV{\ell}(G)\subseteq
	V(G)\setminus W.\] The second chain of inclusions is now an immediate consequence of the latter
	statement and the fact that $\notlcaV{i}(G)=V(G)\setminus\lcaV{i}(G)$ for all $i$.
\end{proof}

We now introduce the main definition of this section. It defines a transformation that maps a DAG
$G$ to a DAG $G'$ by first removing, via the $\ominus$-operator, all vertices that are not $k$-lca
vertices for any $k\leq i$, and then removing all shortcuts. As we shall see, this transformation
always yields an $i$-lca-relevant and regular DAG; see also Theorem~\ref{thm:prune-regular}.

\begin{definition}\label{def:regularization}
	For an integer $i\geq 1$, the \emph{$i$-regularization} of a DAG $G$ is $\reg_i(G)\coloneqq
	(G\ominus \notlcaV{i}(G))^-$. In particular, if $G$ is a DAG on $X$, then we define the
	\emph{regularization} of $G$ to be $\reg(G)\coloneqq\reg_{|X|}(G)$.
\end{definition}

\begin{figure}
    \centering
    \includegraphics[width=0.75\linewidth]{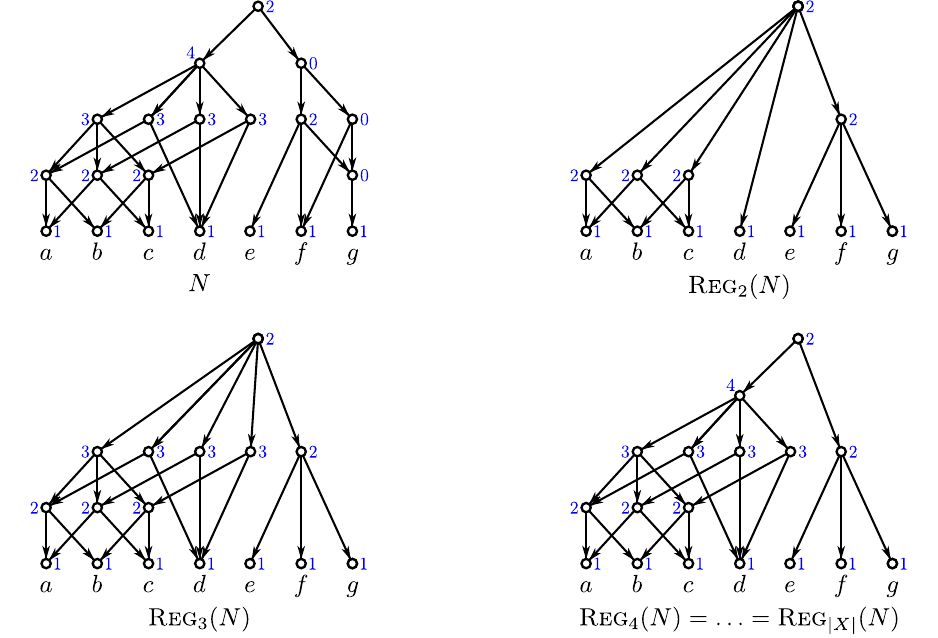}
    \caption{
    A network $N$ on $X=\{a,b,c,d,e,f,g\}$. Numbers $i$ next to a vertex $v$ denote the smallest
    integers such that $v\in \lcaV{i}(N)$. If $i=0$, then the respective vertex is not a $k$-lca
    vertex for any $k$. In addition several networks $\reg_i(N)$ are shown. Note that $\reg_1(N) =
    (X,\emptyset)$ just consists of $X$ and no arcs.
    }
    \label{fig:regi-seq}
\end{figure}

Definition~\ref{def:regularization} is illustrated in Figure~\ref{fig:regi-seq}. The figure
also shows that, for a DAG $G$ on $X$ and integers $j<i$, the DAG $\reg_i(G)$ can be
viewed as a refinement of $\reg_j(G)$, with the sequence starting
at the arc-less DAG $\reg_1(G) = (X,\emptyset)$.

As outlined in the introduction, various approaches to simplify networks have been proposed in
recent years. In \cite{HL:24}, Lindeberg and Hellmuth studied simplifications $G\ominus W$ of
DAGs $G$ in which each vertex $v\notin W$ is an $i$-lca vertex for some
$i\in\mathcal{I}$, where $\mathcal{I}$ is a set of positive integers. The $i$-regularization of
Definition~\ref{def:regularization} is therefore a natural specialization of this
construction to $\mathcal{I}=\{1,\ldots,i\}$, followed by the removal of all shortcuts. In particular, the simplification
$\varphi_{\lca}(G)$ introduced in \cite[Sec.~6]{HL:24} coincides with the ``full'' regularization of
$G$ as given in Definition~\ref{def:regularization}, that is, \[\reg(G) = (G\ominus
\notlcaV{|X|}(G))^- =\varphi_{\lca}(G).\] It follows that $\reg(G)$ satisfies three additional
``simplification'' axioms that are desirable for such transformations, provided that $G$ is a DAG
satisfying the (CL) property; see \cite[Sec.~6]{HL:24} and \cite{Heiss:25} for further details.

\begin{figure}
    \centering
    \includegraphics[width=0.8\linewidth]{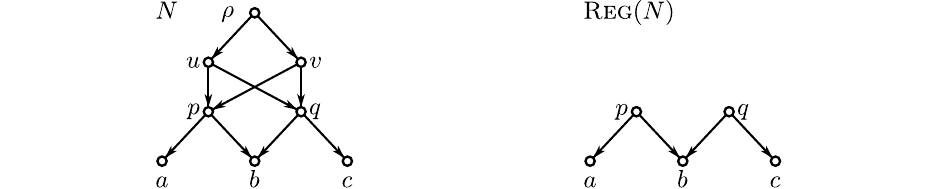}
    \caption{The network $N$ has a unique root, while the corresponding DAG $\reg(N)$ has two.}
    \label{fig:reg_loose_root}
\end{figure}

It is worth noting that, in general, there is no direct relationship between the number of
roots of a DAG $G$ and the number of roots of $\reg_i(G)$ for $i\geq 1$. For example, the network
$N$ in Figure~\ref{fig:reg_loose_root} loses its unique root under regularization: $\reg(N)$ has two
distinct roots. On the other hand, the DAG $G$ in
Figure~\ref{fig:first-hasse-prune} has three roots, while $\reg(G)$ has only one. Thus, even
when $G$ is a network, $\reg_i(G)$ need not be a network, and conversely, a DAG with several roots
may have a regularization that is a network.

Our main motivation for considering $\reg_i(G)$ stems from results established in \cite{HL:24}, which we restate using our notation.

\begin{theorem}[\cite{HL:24}]\label{thm:prune-regular}
	Let $G$ be a DAG on $X$ and $i\geq 1$ be an integer. Then, $\reg_i(G)$ is an $i$-lca-relevant and regular DAG on $X$. Moreover, $G$ is regular if and only if $\reg(G) = G$.
\end{theorem}
\begin{proof}
    Theorem~5.5 in \cite{HL:24} ensures that $G\ominus \notlcaV{i}(G)$ is $i$-lca-relevant for every DAG $G$ and integer $i\geq1$. Since $\reg_i(G)=(G\ominus \notlcaV{i}(G))^-$ is shortcut-free, Observation~\ref{obs:prune-basics} implies that $\reg_i(G)$ is $i$-lca-relevant. Thus, regularity of $\reg_i(G)$ follows from Theorem~\ref{thm:regular-char}. The second statement is an immediate consequence of the first statement and Theorem~\ref{thm:regular-char}.
\end{proof}

The key point of Theorem~\ref{thm:prune-regular} is that $\reg_i(G)$ is regular, which motivates its
name. Consequently, its structure is determined solely by its set of clusters. Although
\cite{HL:24} implicitly considers $\reg_i(G)$, it does not provide a description of $\mathfrak{C}(\reg_i(G))$, which we now establish. To this end, we first introduce a new definition.

\begin{definition}\label{def:lca-clusters}
	Let $G$ be a DAG on $X$. The set of \emph{lca-clusters} of $G$ is defined to be
	\[\mathfrak{C}_{\lca}(G)\coloneqq\{C\in\mathfrak{C}(G)\mid \lca_G(C) \text{ is well-defined}\}.\]
	Furthermore, for each integer $i\in\{1,2,\ldots,|X|\}$, we define
	\[\mathfrak{C}_{\lca}^{i}(G)\coloneqq\{C\in\mathfrak{C}_{\lca}(G)\mid \exists A\in \Xci:\, C
	\text{ is the unique inclusion-minimal element in }\mathfrak{C}(G) \text{ containing }A\}.\]
\end{definition}

Definition~\ref{def:lca-clusters} is illustrated with the help of the following example.
\begin{example*}
	In Figure~\ref{fig:exmpl-LCA}, two networks $N'$ and $N''$ are shown that both have leaf set
	$X=\{a,b,c\}$. For $N'$, we have $\mathfrak{C}(N') = \{\{a,b,c\}, \{a,b\}, \{a\}, \{b\}, \{c\}\}$
	where $\CC_{N'}(u)=\CC_{N'}(v)=\{a,b\}$, i.e., $u$ and $v$ are not distinguishable from their
	associated clusters. Neither $u$ nor $v$ serve as a unique LCA for any subset of leaves. In
	particular, $\CC_{N'}(v)=\CC_{N'}(u)=\{a,b\}\notin \mathfrak{C}_{\lca}(N') $ and we obtain
	$\mathfrak{C}_{\lca}(N') = \{\{a,b,c\}, \{a\}, \{b\}, \{c\}\}$. It is easy to verify that
	$\mathfrak{C}_{\lca}^{1}(N') = \{\{a\}, \{b\}, \{c\}\}$ and $\mathfrak{C}_{\lca}^{3}(N') =
	\mathfrak{C}_{\lca}(N')$. In this example, we even have $\mathfrak{C}_{\lca}^{2}(N') =
	\mathfrak{C}_{\lca}(N')$ since $\{a,b,c\} \in \mathfrak{C}_{\lca}(N')$ is the unique
	inclusion-minimal element in $\mathfrak{C}(N')$ that contains for example $A= \{a,c\} \in
	    \Xci$. For $N''$, we have $\mathfrak{C}(N'')=\{A\mid \emptyset\neq A\subseteq X\}=\Xci[3]$
	and $\mathfrak{C}_{\lca}^{i}(N'') =\Xci$, $1\leq i\leq 3$.
\end{example*}

Note that, by definition, $\mathfrak{C}_{\lca}^{i}(G)\subseteq \mathfrak{C}_{\lca}^{j}(G)$ always
holds whenever $i\leq j$. Moreover, there are DAGs $G$, for example phylogenetic trees, for which
$\mathfrak{C}_{\lca}(G) = \mathfrak{C}(G) $ holds. We now describe the set of clusters of
$\reg_i(G)$.

\begin{theorem}\label{thm:clusters-of-prune}
	For each DAG $G$ on $X$ and every $i\in\{1,2,\ldots,|X|\}$,
	$\mathfrak{C}(\reg_i(G))=\mathfrak{C}_{\lca}^{i}(G)$. In particular, since $\reg_i(G)$ is regular,
	$\reg_i(G)\simeq\Hasse(\mathfrak{C}_{\lca}^{i}(G))$.
\end{theorem}
\begin{proof}
  Let $G$ be a DAG on $X$, put $\ell\coloneqq |X|$ and let $i\in\{1,\ldots,\ell\}$. We first note
  that the statements in Observation~\ref{obs:prune-basics} are, in particular, valid for
  $H\coloneqq \reg_i(G) = (G\ominus \notlcaV{i}(G))^-$ and that $V(H)=\lcaV{i}(G)$. Hence, to prove
  the theorem, we show that $\mathfrak{C}(H)=\mathfrak{C}_{\lca}^{i}(G)$.
          
  First, suppose $C\in\mathfrak{C}(H)$. By definition, there exists some $v\in V(H)$ such that
  $\CC_{H}(v)=C$. Recall that, from Observation~\ref{obs:prune-basics}, $\CC_{H}(v)=\CC_{G}(v)$.
  Moreover, since $v\in\lcaV{i}(G)=V(H)$ there is some $A\in \Xci$ such that $v=\lca_G(A)$. The
  latter implies that $\CC_G(\lca_G(A)) = \CC_G(v)=C$. Moreover, since $\lca_G(A)$ is well-defined,
  we can apply Lemma~\ref{lem:lca-clusters} to conclude that $\lca_G(\CC_G(\lca_G(A)))=\lca_G(C)$ is
  well-defined (and equals $v$). In other words, $C\in\mathfrak{C}_{\lca}(G)$. Moreover,
  Lemma~\ref{lem:lca-clusters} also ensures that $C=\CC_G(\lca_G(A))$ is the unique
  inclusion-minimal cluster in $\mathfrak{C}(G)$ that contains $A$. We have thus shown that
  $C\in\mathfrak{C}_{\lca}^i(G)$. In other words,
  $\mathfrak{C}(H)\subseteq\mathfrak{C}_{\lca}^{i}(G)$.

  To see that $\mathfrak{C}_{\lca}^{i}(G)\subseteq\mathfrak{C}(H)$ holds, consider any cluster $C\in
  \mathfrak{C}_{\lca}^{i}(G)$. One easily verifies that there exists some $v\in V(G)$ such that
  $\CC_G(v)=C$ and $v=\lca_G(C)$, see Lemma~\ref{lem:lca-clusters}. We now claim that $v\in V(H)$.
  Since, by assumption, $C\in \mathfrak{C}_{\lca}^{i}(G)$, there exists, by definition, some $A\in
  \Xci$ such that $C$ is the unique inclusion-minimal cluster in $\mathfrak{C}(G)$ containing
  $A$. Since $A\subseteq C$ and $v=\lca_G(C)$, $v$ is a common ancestor of $A$ and
  $\LCA_G(A)\neq\emptyset$. Consider any $v'\in\LCA_G(A)$. Since $A\subseteq\CC_G(v')$ and $C$ is
  the unique inclusion-minimal cluster in $\mathfrak{C}(G)$ containing $A$, we must have
  $C\subseteq\CC_G(v')$. Hence, $v'$ is a common ancestor $C$. The latter together with the fact
  that $v=\lca_G(C)$ is the unique least common ancestor of $C$ implies that $v=\lca_G(C)\preceq_G
  v'$ must hold. As the latter holds for every $v'\in\LCA_G(A)$, we can conclude that
  $\LCA_G(A)=\{v\}$, i.e., that $v=\lca_G(A)$. Since $A\in \Xci$, it follows that
  $v\in\lcaV{i}(G)=V(H)$, as claimed.  
  
  Observation~\ref{obs:prune-basics} then implies that $C=\CC_G(v)=\CC_{H}(v)\in\mathfrak{C}(H)$
  i.e. $\mathfrak{C}_{\lca}^{i}(G)\subseteq\mathfrak{C}(H)$. Finally, to complete the proof,
  observe that Theorem~\ref{thm:prune-regular} implies that $\reg_i(G)$ is regular. By definition of
  regular DAGs, $\reg_i(G)\simeq\Hasse(\mathfrak{C}_{\lca}^{i}(G))$. 
\end{proof}

Theorem~\ref{thm:clusters-of-prune} provides a characterization of the structure of $\reg_i(G)$ in
terms of certain clusters of $G$. In the rest of this section, we shall pinpoint consequences of
this result and explore further properties of the $i$-regularization operator.

A natural starting point is to consider the regularization of the
archetypal phylogenetic structures, namely trees. One easily verifies
that, in a tree, a vertex $v$ is a $k$-lca vertex for some $k$ if and only if $v$ does not have
out-degree one. Hence, $\notlcaV{i}(T)$ is precisely
the set of vertices of out-degree one in any tree $T$ and for all $i\geq 2$. Moreover, there is a
well-known bijective correspondence between phylogenetic trees on $X$ and hierarchies on $X$, see
e.g.\ \cite[Thm.~3.5.2]{sem-ste-03a}, Thus, Theorem~\ref{thm:clusters-of-prune} immediately
characterizes those graphs $G$ for which $\reg_i(G)$ is a tree. We collect
these observations in the following statement.

\begin{observation}
	For all $i\geq2$, the $i$-regularization $\reg_i(T)$ of a tree $T$ on $X$ coincides with $T'$,
	where $T'$ is the tree obtained from $T$ by suppressing all vertices of out-degree one. In
	particular, $\reg_i(T)=T$ for all phylogenetic trees $T$ and all $i\geq 2$. Moreover, for any DAG
	$G$ and $i\geq2$, the graph $\reg_i(G)$ is a tree if and only if $\mathfrak{C}_{\lca}^{i}(G)$ is a
	hierarchy.
\end{observation}

Another straightforward observation is that LCAs are unaffected by the presence or absence of
shortcuts. Consequently, adding or removing shortcuts in $G$ does not affect $\reg_i(G)$. This can
be formalized as follows.

\begin{lemma}\label{lem:basic-properties-ireg}
  For every DAG $G$ and integer $i\geq1$, $\reg_i(G)=\reg_i(G^-)$.
\end{lemma}
\begin{proof}
	Recall that, by definition, $V(G)=V(G^-)$ and that Lemma~\ref{lem:properties-SF-G} states that
	$u\preceq_G v$ if and only if $u\preceq_{G^-}v$ for all $u,v\in V(G)$. It is now straightforward
	to verify that, for every non-empty $A\subseteq X$, $\LCA_G(A)=\LCA_{G^-}(A)$ and, therefore,
	$\lcaV{i}(G)=\lcaV{i}(G^-)$. This together with Lemma~\ref{lem:ominus-basics} implies that
	$V(\reg_i(G))=\lcaV{i}(G)=V(\reg_i(G^-))$. Since $\reg_i(G)$ and $\reg_i(G^-)$ are shortcut-free
	the latter, taken together with the last statement of Lemma~\ref{lem:same-prec-so-G=H},
	implies that $\reg_i(G)=\reg_i(G^-)$.
\end{proof}

\begin{figure}
    \centering
    \includegraphics[width=0.75\linewidth]{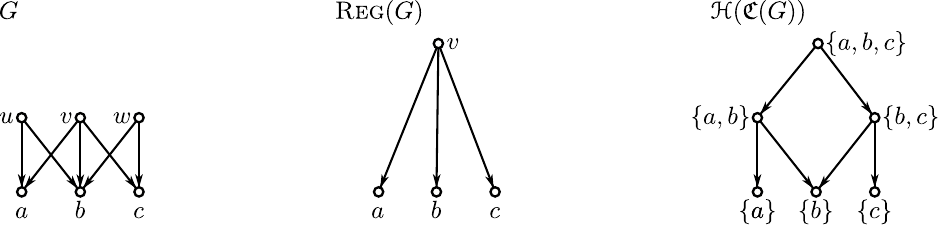}
    \caption{A DAG $G$ for which $\Hasse(\mathfrak{C}(G)) \not\simeq \reg(G)$ holds.
    }
    \label{fig:first-hasse-prune}
\end{figure}

Note that $\Hasse(\mathfrak{C}(G))$ and $\reg(G)$ may differ; see
Figure~\ref{fig:first-hasse-prune}. In this example, for the DAG $G$ we have $v=\lca_G(a,c)$, while
neither $u$ nor $w$ is the unique LCA of any subset of the leaves. As a result, $\reg(G)$ is just a
``star tree''. In contrast, $\Hasse(\mathfrak{C}(G))$ not only differs substantially from $G$ but
also from $\reg(G)$. As the next results shows, $\Hasse(\mathfrak{C}(G))$ and $\reg(G)$ are
isomorphic precisely if $\mathfrak{C}(G)=\mathfrak{C}_{\lca}(G)$.

\begin{theorem}\label{thm:reg_vs_hasse}
	For every DAG $G$, $\reg(G)\simeq\Hasse(\mathfrak{C}_{\lca}(G))$. Moreover,
	$\reg(G)\simeq\Hasse(\mathfrak{C}(G))$ if and only if $G$ satisfies property (CL).
\end{theorem}
\begin{proof}
	Let $G$ be a DAG on $X$ and put $\ell=|X|$. Theorem~\ref{thm:clusters-of-prune} states that
	$\mathfrak{C}(\reg(G))=\mathfrak{C}(\reg_\ell(G))=\mathfrak{C}_{\lca}^{\ell}(G)$. By definition
	and Lemma~\ref{lem:lca-clusters}, $\mathfrak{C}_{\lca}^{\ell}(G)=\mathfrak{C}_{\lca}(G)$ holds.
	Since, by Theorem~\ref{thm:prune-regular}, $\reg(G)$ is regular, the latter two arguments ensure
	that $\reg(G)\simeq\Hasse(\mathfrak{C}_{\lca}(G))$ for all DAGs $G$. The second statement now
	follows immediately from the fact that, by definition, $G$ satisfies property (CL) if and only if
	$\lca_G(\CC_G(v))$ is well-defined for all $v\in V(G)$, which is, if and only if
	$\mathfrak{C}(G)=\mathfrak{C}_{\lca}(G)$.
\end{proof}

We close this section by discussing some computational aspects of $i$-regularization. As a starting
point, note that the property of being $2$-lca-relevant can be tested in polynomial time. More
generally, for every fixed constant integer $k$, testing whether a DAG is $k$-lca-relevant can be
done in polynomial time; see \cite[Obs.~7.6]{HL:24}. Moreover, \cite[Cor.~3.12]{HL:24} gives a
polynomial-time algorithm for testing whether a given DAG is lca-relevant. In contrast, testing
whether a DAG $G$ is $k$-lca-relevant when $k$ is part of the input is NP-complete
\cite[Thm.~7.5]{HL:24}. 
We now extend these results to the computation of $\reg_i(G)$ for certain values of $i$.

\begin{proposition}\label{prop:polytime-regN}
	Given a DAG $G = (V,E)$ on $X$, the DAG $\reg(G)$ can be constructed in polynomial time. Moreover,
	if $i$ is a fixed constant, then $\reg_i(G)$ can be constructed in polynomial time in $|V|$.
\end{proposition}
\begin{proof}
\end{proof}
\begin{proof}
	Let $G=(V,E)$ be a DAG on $X$. By definition, $\reg(G) = (G\ominus \notlcaV{|X|}(G))^-$ and
	$\notlcaV{|X|}(G) = \{v\in V \mid v \text{ is not a } k\text{-lca vertex for any } k\leq |X|\}$.
	Together with Lemma~\ref{lem:lca-clusters}, this implies that $\notlcaV{|X|}(G) = \{v\in V \mid
	v\neq \lca_G(\CC_G(v))\}$. For each vertex $v\in V$, the set $\LCA_G(\CC_G(v))$ can be
	computed in polynomial time using the algorithm from \cite{HL:24}. Hence, we can determine in
	polynomial time whether $v=\lca_G(\CC_G(v))$, and thus construct $\notlcaV{|X|}(G)$ in polynomial
	time. Finally, $\reg(G)=(G\ominus \notlcaV{|X|}(G))^-$ can be computed in $O(|V|^2)$ time; see
	\cite[Section~5]{HL:24} for further details.

	Now let $i$ be a fixed constant. We can determine $\notlcaV{i}(G)$ by computing $\LCA_G(A)$ for
	all subsets $A\subseteq X$ with $|A|\leq i$. There are $O(|X|^i)$ such subsets, and each
	corresponding LCA set can be computed in polynomial time using the algorithm from \cite{HL:24}.
	Since $i$ is fixed, this yields a polynomial-time algorithm for constructing $\notlcaV{i}(G)$. As
	above, $\reg_i(G)=(G\ominus \notlcaV{i}(G))^-$ can then be computed in polynomial time.
\end{proof}

\subsection{Regularizing Normal Networks}
\label{sec:norm}

In this section, we investigate how LCA-based regularization acts on normal networks and, in
particular, clarify the relationships between normality, strong normality, and lca-relevance. We
show that regularization preserves normality and characterize precisely when $\reg(N)=N$; see
Theorem~\ref{thm:normal=>Lrsf(N)-network}. These structural results also provide important
ingredients for Section~\ref{sec:normreg}, where we compare regularization with normalization within
the common $\ominus$-framework.

\begin{figure}[t]
    \centering
    \includegraphics[width=0.75\linewidth]{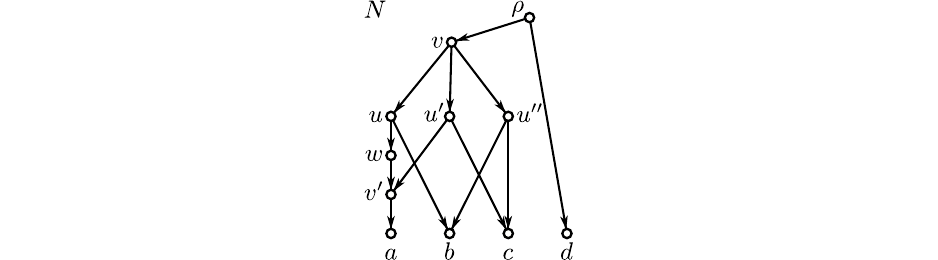} \bigskip
    \begin{tabular}{l|c|c|c|c|c|c}
        & $a,b,c,d$ & $\rho$ & $u,u',u''$ & $v$ & $v'$ & $w$ \\
        \hline
        visible & Yes & Yes & No & Yes & Yes & No\\
        \hline
        $k$-lca & Yes {\footnotesize ($k=1$)} & Yes {\footnotesize ($k=2,3,4$)} & Yes {\footnotesize ($k=2$)} & Yes {\footnotesize ($k=3$)} & No & No \\
    \end{tabular}
    \caption{
    	A network $N$ on $\{a,b,c,d\}$. The table specifies which of the vertices of $N$ are
    	(non-)visible or (non-)$k$-lca vertices. 
    }
    \label{fig:lca-vs-visible}
\end{figure}

We note that the property of a vertex being visible and being a $k$-lca vertex are two unrelated
properties; neither of them implies the other. In fact, there are networks $N$ containing vertices
of every possible combination of (non-)visible and (non-)$k$-lca vertices, see $N$ and the table in
Figure~\ref{fig:lca-vs-visible}. Nevertheless, as the next result shows, every strongly normal
network is $k$-lca-relevant for $k\geq 2$. Thus, all vertices are $k$-lca vertices and, by
Theorem~\ref{thm:char-normal-vis}, at the same time visible. This result has, in parts, been proven
by Willson \cite[Thm~3.9]{Willson2010}, albeit in somewhat different setting and using alternative
notation. 

\begin{lemma}\label{lem:stronglynormal=>2-lca-rel}
	Every strongly normal DAG on $X$ is $k$-lca-relevant for all $k\in \{2,\dots,|X|\}$.
\end{lemma}
\begin{proof}
	Since the proof is straightforward, we have placed it in Appendix~\ref{sec:appx}.
\end{proof}

We now provide some characterizations of strongly normal networks.

\begin{proposition}\label{prop:strong-normal=>2lcarel}
	Let $N$ be a network. Then, the following four statements are equivalent. 
 \begin{enumerate}
    \item $N$ is strongly normal. 
    \item $N$ is normal and 2-lca-relevant.
    \item $N$ is normal and lca-relevant.
    \item $N$ is normal, phylogenetic and does not contain hybrids having only one child.
\end{enumerate}
\end{proposition}
\begin{proof}
	Lemma~\ref{lem:stronglynormal=>2-lca-rel} shows that Condition (1) implies (2), and the
	implication from (2) to (3) is immediate by definition. If $N$ is normal and lca-relevant, then the
	contraposition of Lemma~\ref{lem:outdeg1-no-lca} implies that no vertex of $N$ has out-degree one.
	Hence, Condition~(3) implies~(1).
    The equivalence between (1) and (4) follows from the fact that, in a phylogenetic network, only hybrid vertices can have out-degree one.
\end{proof}

We now show that, for a normal DAG $G$, normality is preserved
under the $\ominus$-operator when applied to vertices of out-degree one, and that these vertices
are precisely those contained in $\notlcaV{i}(G)$.

\begin{proposition}\label{prop:lca-i/reg-i_of_normal}
	Let $G$ be a normal DAG on $X$ and $W\subseteq U\coloneqq\{v\in V(G)\mid \outdeg_G(v)=1\}$. Then,
	$G\ominus W$ is a normal DAG on $X$ and $G\ominus U$ is strongly normal.
	In particular, $\notlcaV{i}(G)=U$ for all integers $i\geq 2$.
\end{proposition}
\begin{proof}
Since the proof is straightforward but technical, we have placed it in Appendix~\ref{sec:appx}.
\end{proof}

\begin{figure}
    \centering
    \includegraphics[width=0.75\textwidth]{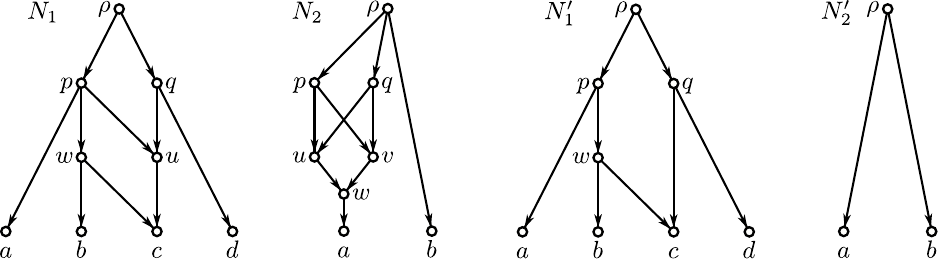}
    \caption{
    	Two networks $N_1$ and $N_2$ for which the regularization and normalization coincide, i.e.,
    	$\reg(N_1)=N_1'=\norm(N_1)$ and $\reg(N_2)=N_2'=\norm(N_2)$.
    }
    \label{fig:outdeg1}
\end{figure}

Note that normal networks are not the only type of DAGs for which $\notlcaV{i}(N)=\{v\in V(N)\mid
\outdeg_N(v)=1\}$. For example, consider the network $N_1$ in Figure~\ref{fig:outdeg1}, for which
$\notlcaV{i}(N_1)=\{u\}$ for all $i\geq2$, i.e., $\notlcaV{i}(N_1)$ contains precisely the only
vertex with out-degree one. Here, the vertex $u$ has no tree-child and thus $N_1$ is not normal.
Nevertheless, $N_1$ is phylogenetic and shortcut-free.

A straightforward consequence of Proposition~\ref{prop:lca-i/reg-i_of_normal} is the following. 
For a normal DAG $G$, the graph $G\ominus \notlcaV{2}(G)$ is already shortcut-free.
Moreover, $\notlcaV{i}(G)=\notlcaV{2}(G)$ for all $i\geq 2$, and hence
$\reg_i(G)$ is obtained simply by removing the vertices of out-degree one. Hence, we obtain

\begin{corollary}\label{cor1:lca-i/reg-i_of_normal}
	For every normal DAG $G$ it holds that $\reg(G) = \reg_i(G) = G\ominus \notlcaV{2}(G)$ for all
	integers $i\geq 2$, and, in particular, $V(\reg(G))=\lcaV{2}(G)$.
\end{corollary}

Recall that $\reg(N)$ may have multiple roots, even though $N$ itself is a network (c.f.\
Figure~\ref{fig:reg_loose_root}). Nevertheless, the latter cannot happen if $N$ is a normal network
as we now show in the main result of this section.

\begin{theorem}\label{thm:normal=>Lrsf(N)-network}
    Let  $N$ be a normal network on $X$. Then, $\reg(N)$ is a normal phylogenetic network on $X$.
    In particular,  the following statements are equivalent.
    \begin{enumerate}
        \item $\reg(N) = N$.
        \item $N$ is regular.
        \item $N$ is strongly normal.
        \item $\reg_i(N) = N$ for all $i \in \{2,\dots, |X|\}$.
    \end{enumerate}
\end{theorem}
\begin{proof}
    Let $N$ be a normal network on $X$ and let $\rho$ be its unique root. By
    Corollary~\ref{cor1:lca-i/reg-i_of_normal} and Proposition~\ref{prop:lca-i/reg-i_of_normal} we
    have that $\reg(N)=N\ominus \notlcaV{2}(N)$, where $\notlcaV{2}(N)=\{v\in V(N)\mid
    \outdeg_N(v)=1\}$.  
    This together with Proposition~\ref{prop:lca-i/reg-i_of_normal} implies that $\reg(N)$ is a normal DAG on $X$. We show that $\reg(N)$ is a network, i.e,
    that $\reg(N)$ has a unique root.

    If $\rho\in V(\reg(N))$, then $v\preceq_N \rho$ for all $v\in V(N)$. Observation~\ref{obs:prune-basics} 
    therefore implies $v \preceq_{\reg(N)} \rho$,
    for all $v\in V(\reg(N))$ and hence, $\rho$ remains the unique root of $\reg(N)$. 
    Now suppose that $\rho\not\in V(\reg(N))$, that is, $\outdeg_N(\rho) = 1$. 
    Starting at $\rho$, follow the unique outgoing arc until reaching the first vertex $v$ with
    $\outdeg_N(v)\neq 1$; such a vertex exists and may be a leaf. Then $v$ is the unique
    $\preceq_N$-maximal vertex of $N$ whose out-degree is distinct from one. By
    Observation~\ref{obs:prune-basics}, $u\preceq_{\reg(N)}v$ for every $u\in V(\reg(N))$. Thus, $v$
    is the unique root of $\reg(N)$. In summary, $\reg(N)$ is a normal \emph{network} on $X$.
    Moreover, Theorem~\ref{thm:prune-regular} implies that $\reg(N)$ is regular and thus, by
    \cite[Thm~4.6]{HL:24}, also phylogenetic. 
		
    It remains to prove the equivalence of Statements~(1) through (4).
    By Theorem~\ref{thm:prune-regular}, Statements~(1) and (2) are equivalent. 
    Suppose that (1) holds. Since $N$ is shortcut-free, $N^- = N =\reg(N)$. This together with
    Theorem~\ref{thm:prune-regular} implies that $N$ is lca-relevant. Since $N$ is normal,
    Proposition~\ref{prop:strong-normal=>2lcarel} implies that $N$ is strongly normal. Hence, (1)
    implies (3). Suppose now that~(3) holds.
    Lemma~\ref{lem:stronglynormal=>2-lca-rel} implies that $N$ is $i$-lca-relevant for all $i \in
    \{2,\dots, |X|\}$. Thus, $\notlcaV{i}(N)=\emptyset$ and since $N$ is shortcut-free,
    we obtain $N = N^-=(N\ominus\notlcaV{i}(N))^-=\reg_{i}(N)$ for all $i \in \{2,\dots,
    |X|\}$. Hence, (3) implies (4). Finally, since $\reg(N) = \reg_{|X|}(N)$, Statement (4) immediately implies (1), which completes this proof. 
\end{proof}

Note, in particular, that Statement (2) and (4) of Theorem~\ref{thm:normal=>Lrsf(N)-network} are not
necessarily equivalent in case $N$ is not normal. For example, the network $N''$ in
Figure~\ref{fig:exmpl-LCA} is regular, but $\reg_2(N'')=N''\ominus \rho\neq N''$.

\section{Normalizing versus Regularizing}
\label{sec:normreg}

In \cite{francis2021normalising} the concept of \emph{normalization} was introduced, that is, the
transformation of a given network $N$ into a normal version $\norm(N)$ of the network that aims to
retain the underlying evolution information whilst producing a network that has the desirable
features of a normal network. Since clustering systems naturally arise in network inference,
$\reg(N)$ -- by virtue of being regular -- also encodes useful phylogenetic signals through
leaf-based information, since every vertex is the unique LCA of some subset of leaves (cf.\
Theorem~\ref{thm:prune-regular}). It is therefore natural to investigate how $\reg(N)$ and
$\norm(N)$ are related, as this could be useful information in practice when using these techniques.

One of the main goals of this section is to characterize those networks for which $\reg(N)=
\norm(N)$ holds. In particular, we show that $\norm(N)$ can be fully expressed as $(N\ominus W)^-$
for a particular subset $W\subseteq V(N)$ and thus, conceptually, falls within the same paradigm as
$\reg(N)$. Before providing the definition of $\norm(N)$, we emphasize that (implicitly), all
networks in \cite{francis2021normalising} are separated, i.e., all hybrids $v$ in $N$ have
$\outdeg_N(v)=1$. Moreover, no vertex $v$ of the networks used in \cite{francis2021normalising}
satisfies $\indeg(v)=\outdeg(v)=1$. Hence, they are phylogenetic ``up to'' possibly the root having
out-degree 1. We thus slightly generalize the definition of $\norm(N)$ to arbitrary networks $N$.

\begin{definition}[{\cite[Sec.~2]{francis2021normalising}}]\label{def:normalisation}
	Let $N$ be a network. The directed graphs $\cov(N)$ and $\norm(N)$ are constructed as follows. 
 \begin{enumerate}
        \item Let $G = (\vis(N),\,A)$ be the DAG whose arc set $A$ comprises 
              all arcs $u\to v$ for which $u,v\in \vis(N)$ and $v\prec_N u$ holds
        \item Put $\cov(N) = G^-$.
        \item $\norm(N)$ is obtained from $\cov(N)$ by suppressing all vertices $v$ of $\cov(N)$ for
              which $\indeg_{\cov(N)}(v)=\outdeg_{\cov(N)}(v)=1$.
    \end{enumerate}
\end{definition}

\begin{figure}
    \centering
    \includegraphics[width=0.75\linewidth]{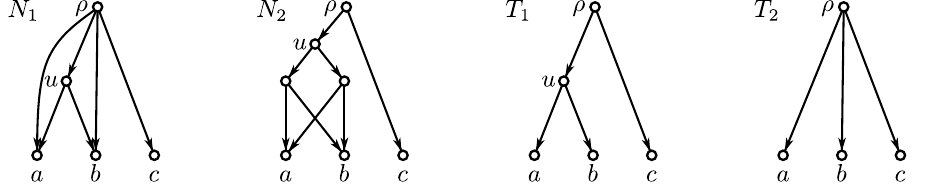}
    \caption{Two networks $N_1$, $N_2$ and two trees $T_1$, $T_2$ on $\{a,b,c\}$, for which $T_1
             =\reg(N_1) = \norm(N_2)$ and $T_2=\norm(N_1)=\reg(N_2)$.}
    \label{fig:normal-prune-intro}
\end{figure}

There are many networks $N$ for which the normalization $\norm(N)$ and the regularization $\reg(N)$
of $N$ coincide. For example, considering the networks shown in Figure~\ref{fig:outdeg1}, we have
both $N_1'=\norm(N_1)=\reg(N_1)$ and $N_2'=\norm(N_2)=\reg(N_2)$. However, this is not always the
case. For example, see Figure~\ref{fig:normal-prune-intro} where $\norm(N_1)=\reg(N_2)=T_2\neq T_1
=\reg(N_1) = \norm(N_2)$. In particular, there are lca-vertices that are not visible (e.g.\ vertex
$u$ in $N_1$) and there are visible vertices that are not lca-vertices (e.g.\ the vertex $u$ in
$N_2$). In other words, $\reg(N)$ and $\norm(N)$ may identify different vertices as ``omittable''. A
technical difficulty in comparing regularization and normalization is that the operator $\reg$ is
invariant under the addition or removal of shortcuts. More precisely, $\reg(N^-)=\reg(N)$ always
holds; cf.\ Lemma~\ref{lem:basic-properties-ireg}(2). In contrast, the normalization operator
$\norm$ does not share this invariance property. 

To illustrate this fact, consider the networks shown in Figure~\ref{fig:normal-prune-intro}. The
network $N_1$ differs from $T_1$ only by the presence of the two shortcuts $\rho \to a$ and
$\rho \to b$. Although $\reg(N_1^-)=\reg(N_1)=T_1$, we have $\norm(N_1)=T_2 \neq T_1 =
\norm(N_1^-)$. Thus, in general, $\norm(N) \neq \norm(N^-)$. Nevertheless, we show that
normalization can be characterized using the $\ominus$-operator together with the removal of
shortcuts. To this end, we start with establishing the connection between the graph $\cov(N)$ in
Definition~\ref{def:normalisation} and the $\ominus$-operator combined with shortcut removal.

\begin{lemma}\label{lem:cov(N)-ominus-form}
    For all networks $N$ on $X$, $\cov(N)=(N\ominus \notvis(N))^-$ is a normal network on $X$. 
\end{lemma}
\begin{proof}
	Since the proof is straightforward, we have placed it in Appendix~\ref{sec:appx}.
\end{proof}

To express $\norm(N)$ in terms of the $\ominus$-operator, we need to determine the set $W$ for which
$\norm(N)=(N\ominus W)^-$ holds. To this end, we define the following set 
\[
\dvis(N) \coloneqq \{v \in \vis(N) \mid \indeg_{\cov(N)}(v) = \outdeg_{\cov(N)}(v) = 1\}.
\]
The vertices in $\dvis(N)$ are called \emph{dispensably visible} (or \emph{$\wvis$}) vertices. By
definition, $\norm(N) = \cov(N)$ if and only if $\cov(N)$ contains no vertices $u$ with
$\indeg_{\cov(N)}(u) = \outdeg_{\cov(N)}(u) = 1$, that is, if and only if $N$ contains no $\wvis$
vertices. Importantly, the $\wvis$ vertices in $\vis(N)$ are determined in terms of their properties
in $\cov(N)$. Nevertheless, as the next result shows $\wvis$ vertices of $N$ can be identified
solely based on properties they satisfy in $N$ itself.

\begin{lemma}\label{lem:parents/children-in-cov}
	Let $N$ be a network on $X$. We define the following two sets for all vertices $v\in \vis(N)$:
    \begin{align*}
    P_{\visop}(v) &\coloneqq\{u\in \vis(N)\setminus \{v\}\mid \text{there is some } uv\text{-path } P \text{ in } N \text{ s.t. }  V(P)\cap \vis(N)=\{u,v\}\} \\ 
    C_{\visop}(v) &\coloneqq\{u\in \vis(N)\setminus \{v\}\mid \text{there is some } vu\text{-path } P \text{ in } N \text{ s.t. }  V(P)\cap \vis(N)=\{u,v\}\}  
    \end{align*}
    Then, for all $v\in \vis(N)$, the following two statements are equivalent. 
    \begin{enumerate}
        \item $v\in \dvis(N)$.
        \item (a) $P_{\visop}(v)\neq\emptyset$ and $P_{\visop}(v)$ contains a unique $\preceq_N$-minimal vertex and 
        
        			 (b) $C_{\visop}(v)\neq\emptyset$ and $C_{\visop}(v)$ contains a unique $\preceq_N$-maximal vertex.
    \end{enumerate}
    Moreover, if $v\in \dvis(N)$, then $v\notin\LCA_N(A)$ for all non-empty $A\subseteq X$.
    In particular, $\dvis(N)\subseteq \notlcaV{i}(N)$ for all integers $i\geq2$.
\end{lemma}
\begin{proof}
	Let $N$, $P_{\visop}(v)$ and $C_{\visop}(v)$ be defined as in the lemma's statement for some $v\in
	\vis(N)$. We recall that Lemma~\ref{lem:cov(N)-ominus-form} and Observation~\ref{obs:prune-basics}
	together implies that \begin{equation}\label{eq:prec-equiv-1'} v\preceq_{\cov(N)} u \iff
	v\preceq_N u,\quad\text{ for all $u,v\in \vis(N)=V(\cov(N))$.} \end{equation} First assume that
	$v$ is $\wvis$ and thus, that $\indeg_{\cov(N)}(v)=\outdeg_{\cov(N)}(v)=1$. Let $c$ denote the
	unique child of $v$ in $\cov(N)$. Hence, $v,c\in \vis(N)$ and there is no vertex $u$ of $\cov(N)$
	such that $c\prec_{\cov(N)}u\prec_{\cov(N)} v$. This together with the definition of $\cov(N)$
	implies that $c\prec_N v$ and that there is no $u \in \vis(N)$ such that $c\prec_{N}u\prec_{N} v$.
	The latter two arguments imply that every path $P=v\leadsto c$ in $N$ must satisfy $V(P)\cap
	\vis(N)=\{v,c\}$ and, therefore, $c\in C_{\visop}(v)$. Let $w\in C_{\visop}(v)$ be a vertex such
	that $w\neq c$. By definition of $C_{\visop}(v)$, it holds that $v\neq w$ and there is a
	$vw$-path in $N$ and thus, $w\prec_N v$ and, furthermore, $w\in \vis(N) = V(\cov(N))$ which
	together with Eq.~\eqref{eq:prec-equiv-1'} implies that $w\prec_{\cov(N)} v$. Since $c$ is the
	unique child of $v$ in $\cov(N)$ and $w\neq c$, we therefore have $w\prec_{\cov(N)} c$ which,
	again, is equivalent to $w\prec_N c$. In other words, $c$ must be the unique $\preceq_N$-maximal
	vertex in the set $C_{\visop}(v)$. With analogous arguments, one shows that the unique parent $p$
	of $v$ in $\cov(N)$ is the unique $\preceq_N$-minimal element of $P_{\visop}(v)$.
 
	Assume, for the converse, that $P_{\visop}(v)$ contains a unique $\preceq_N$-minimal vertex $p$
	and that $C_{\visop}(v)$ contains a unique $\preceq_N$-maximal vertex $c$. We first show that if
	an arc $v\to u$ of $\cov(N)$ exist, then $u\in C_{\visop}(v)$. To this end, note first that if
	$v\to u$ is an arc of $\cov(N)$, then, by definition, $u\prec_N v$. Let $P$ denote any path from
	$v$ to $u$ in $N$. If there would exist a visible vertex $w\in V(P)\setminus\{u,v\}$, then
	$u\prec_N w\prec_N v$ and Eq.~\eqref{eq:prec-equiv-1'} would imply that $u\prec_{\cov(N)}
	w\prec_{\cov(N)} v$; but that would make the arc $v\to u$ a shortcut of $\cov(N)$, contradicting
	$\cov(N)$ being shortcut-free. Hence $V(P)\cap \vis(N)=\{u,v\}$ and, therefore, $u\in
	C_{\visop}(v)$. Thus, any arc $v\to u$ that potentially exists in $\cov(N)$ must satisfy $u\in
	C_{\visop}(v)$. Now, by definition, each $u\in C_{\visop}(v)$ satisfies $u\prec_N v$. Hence, for
	each $u\in C_{\visop}(v)$, the arc $v\to u$ is initially added to the DAG $G$ as in
	Definition~\ref{def:normalisation}(1) from which $\cov(N) = G^-$ is obtained from. Since $c$ is
	the unique $\preceq_N$-maximal vertex in $C_{\visop}(v)$ it follows that any other vertex $w\neq
	c$ in $C_{\visop}(v)$ (if there is one), must satisfy $w\prec_N c$. By the latter arguments, the
	arcs $v\to c$, $c\to w$ and $v\to w$ exists in $G$. However, since the arcs $v\to w$ is a shortcut
	in $G$, it will thus not remain in $\cov(N)=G^-$. As this holds for all vertices $w\in
	C_{\visop}(v)\setminus \{c\}$ it follows that only the arc $v\to c$ exists in $\cov(N)$ and
	therefore, that $c$ is the unique child of $v$ in $\cov(N)$. In summary, $\outdeg_{\cov(N)} = 1$.
	Using similar arguments, one can show that the unique $\preceq_N$-minimal vertex $p$ in
	$P_{\visop}(v)$ is the unique parent of $v$ in $\cov(N)$ and thus, that $\indeg_{\cov(N)} = 1$. By
	definition, $v\in \dvis(N)$. Therefore Statement (1) and (2) are equivalent.

	Now consider any $v\in \dvis(N)$. Since $\outdeg_{\cov(N)}(v) =1$, one easily verifies that there
	is a path $P=v\leadsto u$ in $\cov(N)$ where $u$ satisfies $\outdeg_{\cov(N)}(u)\neq 1$, and where
	each vertex $w$ along $P$ distinct from $u$ (if there is one) satisfies $w\in \dvis(N)$. Clearly,
	$\CC_{\cov(N)}(v)=\CC_{\cov(N)}(u)$ and $u\notin \dvis(N)$. By Lemma~\ref{lem:cov(N)-ominus-form},
	the vertices in $X$ remain vertices in $\cov(N)$. This and Eq.~\eqref{eq:prec-equiv-1'} implies
	that, for each $x\in X$, we have $x\preceq_{\cov(N)} u$ if and only if $x\preceq_N u$.
	Consequently, $\CC_{\cov(N)}(u)=\CC_N(u)$. Similarly, $\CC_{\cov(N)}(v)=\CC_N(v)$. Since,
	additionally, $u\prec_N v$, we have shown that $v$ has a proper descendant $u$ in $N$ with the
	same cluster as $v$. Consequently, $v\notin \LCA_N(\CC_N(v))$ must hold, and this together with
	\cite[Lem~3.4]{HL:24} implies that $v\notin\LCA_N(A)$ and, in particular, $v\neq \lca_N(A)$ for
	all non-empty subsets $A\subseteq X$. As the latter holds for all vertices $v\in \dvis(N)$, we
	conclude that $\dvis(N)\subseteq \notlcaV{|X|}(N)$. By Lemma~\ref{lem:lca-sets-inclusions}, we
	thus have $\dvis(N)\subseteq\notlcaV{i}(N)$ for all $i\geq 2$.
\end{proof}

By definition, suppression can be formulated in terms of the $\ominus$-operator (c.f.\ the
discussion right after Definition~\ref{def:ominus}), hence the following is an immediate result of
Lemma~\ref{lem:cov(N)-ominus-form} and the definition of $\norm(N)$.

\begin{theorem}\label{thm:cov(N)-ominus-form}
 For all networks $N$ we have $\norm(N)=\cov(N)\ominus \dvis(N)=(N\ominus \notvis(N))^-\ominus
 \dvis(N)$.
\end{theorem}

\begin{figure}
    \centering
    \includegraphics[width=0.75\linewidth]{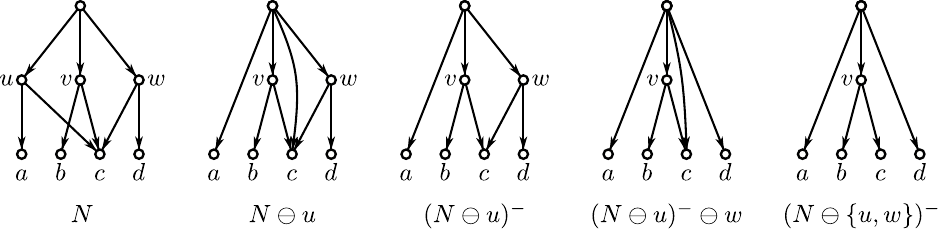}
    \caption{Some networks resulting from the $N$ on the very left after application of the
             $\ominus$-operator and removal of shortcuts. In particular, note that $(N\ominus
             u)^-\ominus w$ is not isomorphic to $(N\ominus \{u,w\})^-$.} \label{fig:ominus-exmpls}
\end{figure}

Note that $(N\ominus U)\ominus W = N\ominus(U\cup W)$ always holds by
Lemma~\ref{ominus:commutative}. However, the $\ominus$-operator could result in networks $N\ominus
U$ that contain shortcuts even if the original network $N$ is shortcut-free; for example, see the
network $N\ominus u$ in Figure~\ref{fig:ominus-exmpls}. Hence, in general, $(N\ominus U)^-\ominus
W\neq(N\ominus(U\cup W))^-$, since $(N\ominus U)^-\ominus W$ may contain shortcuts that, clearly, do
not appear in $(N\ominus(U\cup W))^-$; see Figure~\ref{fig:ominus-exmpls} again for an example.
Intriguingly, as we shall see in Theorem~\ref{thm:norm(N)-ominus-form}, $\norm(N)=(N\ominus
(\notvis(N)\cup \dvis(N)))^-$. In order to prove this theorem and for completeness, we now show that
Theorem~3.3 of \cite{francis2021normalising} also holds for the more general setting of networks
that are not necessarily phylogenetic and separated.

\begin{proposition}
\label{prop:normN-basics}
    
    Let $N$ be a network on $X$. The following statements hold.
    \begin{enumerate}
    \item $\norm(N)$ is a normal network on $X$ for which $V(\norm(N))=\vis(N)\setminus \dvis(N)
              = V(N)\setminus (\notvis(N)\cup \dvis(N))$. 
    \item For all $u,v\in V(\norm(N))$, $v\preceq_{\norm(N)} u \iff v\preceq_N u$.
    \item If $v\in V(\norm(N))$, then $\CC_{\norm(N)}(v)=\CC_N(v)$.
    \item No vertex $v$ of $\norm(N)$ satisfy $\indeg_{\norm(N)}(v)=\outdeg_{\norm(N)}(v)=1$. In
          particular, $\norm(N)$ is phylogenetic if and only if its root has out-degree distinct
          from 1. \item $\norm(N)=N$ if and only if $N$ is normal and there exists no $v\in V(N)$
          such that $\indeg_N(v)=\outdeg_N(v)=1$.
    \item If $N$ is separated, then $\norm(N)$ is separated.
    \item $\norm(N)$ is a tree if and only if $\{\CC_N(v)\mid v\in \vis(N)\}$ is a hierarchy.
\end{enumerate} 
\end{proposition}
\begin{proof}
	Since the proof is straightforward and resembles \cite[Thm.~3.3]{francis2021normalising}, we have
	placed it in Appendix~\ref{sec:appx}.
\end{proof}

We are now in the position to extend Theorem~\ref{thm:cov(N)-ominus-form} and to characterize those
networks $N$ for which $\norm(N)$ and $\reg(N)$ coincide.
\begin{theorem}
\label{thm:norm(N)-ominus-form}
	Let $N$ be a network. Then $\norm(N)=(N\ominus (\notvis(N)\cup \dvis(N)))^-$. Moreover, for each
	$i\geq2$, we have
\begin{alignat}{3}
    &V(\reg_{i}(N))     \cap  V(\norm(N)) &&= \lcaV{i}(N) \cap \vis(N) \label{eq:intersect}\\ 
    &V(\reg_{i}(N)) \setminus V(\norm(N)) &&= \lcaV{i}(N) \setminus \vis(N) \label{eq:setminus}\\
    &V(\norm(N)) \setminus V(\reg_{i}(N)) &&=\vis(N) \setminus (\lcaV{i}(N)\cup \dvis(N)). \label{eq:setminus2}
\end{alignat}
In particular, $\norm(N)=\reg_{i}(N)$ if and only if $\lcaV{i}(N) = \vis(N)\setminus \dvis(N)$ and,
therefore, $\norm(N)=\reg(N)$ if and only if $\lcaV{|X|}(N) = \vis(N)\setminus \dvis(N)$.
\end{theorem}
\begin{proof} 
	Let $N$ be a network on $X$. For simplicity, put $N'\coloneqq(N\ominus (\notvis(N)\cup
	\dvis(N)))^-$. By definition of $\notvis(N)$ and $\dvis(N)$, $X\subseteq V(N)\setminus
	(\notvis(N)\cup \dvis(N))$. This together with Observation~\ref{obs:prune-basics} ensures that
	$N'$ is a shortcut-free DAG on $X$ such that $V(N')=(V(N)\setminus \notvis(N))\setminus
	\dvis(N)=\vis(N)\setminus \dvis(N)$. In addition, Proposition~\ref{prop:normN-basics} shows, in
	particular, that $\norm(N)$ is a shortcut-free network on $X$ with $V(\norm(N))=\vis(N)\setminus
	\dvis(N)$. To summarize so far, the graphs $\norm(N)$ and $N'$ have the same set of vertices and
	are both shortcut-free. Moreover, by Observation~\ref{obs:prune-basics}, we have $u\preceq_{N'} v$
	if and only if $u\preceq_N v$ for all $u,v\in V(N')$. By Proposition~\ref{prop:normN-basics},
	$u\preceq_{\norm(N)} v$ if and only if $u\preceq_N v$ for all $u,v\in V(\norm(N))$. Combining the
	latter two arguments allow us to conclude that, for all $u,v\in \vis(N)\setminus \dvis(N)$, we
	have $u\preceq_{N'} v$ if and only if $u\preceq_{\norm(N)} v$. By
	Lemma~\ref{lem:same-prec-so-G=H}, we thus have $\norm(N)=N'=(N\ominus (\notvis(N)\cup
	\dvis(N)))^-$.
	
	Now, fix some integer $i\geq2$. By definition, we have $V(\reg_i(N))= \lcaV{i}(N)$. By
	Proposition~\ref{prop:normN-basics}, $V(\norm(N))=\vis(N)\setminus \dvis(N)$. In addition,
	Lemma~\ref{lem:parents/children-in-cov} implies that $\dvis(N)\subseteq \notlcaV{i}(N)$ and,
	therefore, $\lcaV{i}(N)\cap \dvis(N) = \emptyset$. It now readily follows that $V(\norm(N)) \cap
	V(\reg_i(N)) = \vis(N)\cap \lcaV{i}(N)$. Hence, Equation~\eqref{eq:intersect} holds. Furthermore,
	$V(\reg_i(N))\setminus V(\norm(N)) = \lcaV{i}(N) \setminus (\vis(N)\setminus \dvis(N))$. Since
	$\lcaV{i}(N)\cap \dvis(N) = \emptyset$, the latter simplifies to $V(\reg_i(N))\setminus
	V(\norm(N)) = \lcaV{i}(N) \setminus \vis(N)$, i.e., Equation~\eqref{eq:setminus} holds. Moreover,
	$V(\reg_i(N))= \lcaV{i}(N)$ and $V(\norm(N))=\vis(N)\setminus \dvis(N)$ implies that
	Equation~\eqref{eq:setminus2} holds.
			  
	To see that the last statement in the theorem holds, first assume that $\norm(N)=\reg_i(N)$. Since
	$V(\norm(N))=\vis(N)\setminus \dvis(N)$ and $V(\reg_i(N))=\lcaV{i}(N)$, we must in particular have
	$\vis(N)\setminus \dvis(N)=V(\norm(N))=V(\reg_i(N))=\lcaV{i}(N)$. Conversely, assume that
	$\lcaV{i}(N) = \vis(N)\setminus \dvis(N)$. Recall that $\norm(N)=(N\ominus (\notvis(N)\cup
	\dvis(N)))^-$ and so \begin{align*} \reg_i(N) = (N\ominus \notlcaV{i}(N))^- &= (N\ominus
	(V\setminus \lcaV{i}(N)))^- \\ &= (N \ominus (V\setminus (\vis(N)\setminus \dvis(N))))^- \\ & = (N
	\ominus ((V\setminus \vis(N)) \cup (V \cap \dvis(N))))^- \\ & = (N \ominus (\notvis(N) \cup
	\dvis(N)))^- = \norm(N). \end{align*} The latter, in particular, shows that $\lcaV{i}(N) =
	\vis(N)\setminus \dvis(N)$ implies that $\reg_i(N) = \norm(N)$. Hence setting $i=|X|$ implies that
	$\norm(N)=\reg(N)$ if and only if $\lcaV{|X|}(N) = \vis(N)\setminus \dvis(N)$.
\end{proof}

\begin{corollary}\label{cor:sets}
    For a network $N$ and an integer $i\geq2$, the following statements are equivalent.
\begin{enumerate}
    \item $\norm(N)=\reg_i(N)$
    \item $\vis(N)\setminus \dvis(N)=\lcaV{i}(N)$
    \item $\vis(N)=\lcaV{i}(N)\cupdot \dvis(N)$
    \item $\notvis(N)\cupdot \dvis(N)=\notlcaV{i}(N)$
    \item $\lcaV{i}(N)\cap \notvis(N)=\emptyset$ and $\notlcaV{i}(N)\cap \vis(N)=\dvis(N)$.
\end{enumerate}
\end{corollary}
\begin{proof}
Since the proof is tedious but straightforward, we have placed it in Appendix~\ref{sec:appx}.
\end{proof}

\begin{corollary}
   Given a network $N$, testing if $\reg(N)=\norm(N)$ can be done in polynomial time in $|V(N)|$.
\end{corollary}
\begin{proof}
	Let $N$ be a network on $X$. As outlined in the appendix of \cite{francis2021normalising},
	$\norm(N)$ can be computed in polynomial time, whenever the network $N$ is separated. In fact, the
	exact same approach may be used to compute $\norm(N)$ for a network that is not necessarily
	separated which, in particular, provides access to the set $V(N)=\vis(N)\setminus\dvis(N)$. As
	outlined in the proof of Proposition~\ref{prop:polytime-regN}, the set $\notlcaV{|X|}(N)$ and,
	thus, $\lcaV{|X|}(N)= V(N)\setminus \notlcaV{|X|}(N)$ can be computed in polynomial time in
	$|V(N)|$. Since, by Theorem~\ref{thm:norm(N)-ominus-form}, $\norm(N)=\reg(N)$ if and only if
	$\lcaV{|X|}(N) = \vis(N)\setminus \dvis(N)$, we can thus test in polynomial time whether
	$\norm(N)=\reg(N)$ by a straightforward set comparison in $O(|V(N)|^2)$ time.
\end{proof}


\section{Discussion and Outlook}\label{sec:future}

In this paper, we have introduced and explored a framework for simplifying DAGs, where the
simplified DAGs have the form $(G\ominus W)^-$ for certain subsets $W$ of the vertices of $G$. To be
more precise, in Section~\ref{sec:reg}, we studied the $i$-regularization $\reg_i(G)$ of a DAG $G$,
where $W$ is the set $\notlcaV{i}(G)$ of vertices that are not the unique least common ancestor of
any subset of leaves of size at most $i$. In particular, we noted that $\reg_i(G)$ is both
$i$-lca-relevant and regular (c.f. Theorem~\ref{thm:prune-regular}). We then proceeded to describe
the clusters of $\reg_i(G)$ in Theorem~\ref{thm:clusters-of-prune} which, due to regularity, fully
determine the structure of $\reg_i(G)$. In Section~\ref{sec:norm}, we studied regularizations of
normal networks, in particular showing that regularizations preserve the property of being a normal
network (Theorem~\ref{thm:normal=>Lrsf(N)-network}). Finally, in Section~\ref{sec:normreg}, we
studied the normalization $\norm(N)$ of a network $N$ and showed, in
Theorem~\ref{thm:norm(N)-ominus-form}, that $\norm(N)=(N\ominus W)^-$ for $W=\notvis(N)\cup
\dvis(N)$. In other words, the set $W$ contains all vertices of $N$ that are not visible and all
vertices of $N$ that are visible but have in- and out-degree 1 in $\cov(N)$ (where the latter
vertices, in particular, can be determined by properties in $N$ itself, see
Lemma~\ref{lem:parents/children-in-cov}). Moreover, Theorem~\ref{thm:norm(N)-ominus-form} also
characterizes those networks $N$ for which $\norm(N)$ and $\reg_i(N)$ coincide.

The comparison between regularization and normalization also highlights that neither simplification
is uniformly more informative than the other. Their behavior depends on which type of leaf-based
information is regarded as relevant. There are regular networks for which regularization retains
essentially the entire structure while normalization retains only very little, since most vertices
are not visible. Conversely, there are networks for which normalization retains substantial
structure although few, or no, internal vertices occur as unique LCAs of subsets of leaves, so that
regularization removes most of the network (refer back to Figure~\ref{fig:normal-prune-intro},
from which larger examples can easily be obtained). Thus, the two approaches should be viewed as
complementary rather than competing simplification procedures. This difference also comes with a
potential complexity trade-off. Regular DAGs may in general contain many more vertices than normal
or strongly normal networks (see, for example, \cite[Thm.~5.1]{Willson2010}), and hence a
regularization that retains a large collection of LCA-induced clusters may be substantially more
complex than a corresponding normalization. On the other hand, such additional vertices may be
relevant when they are supported by the available data, for example through clusters or inferred LCA
relationships. Which simplification is preferable therefore depends not only on structural
considerations but also on the type of information that can be reliably inferred.

It could be worthwhile to investigate how different choices of subsets $W$ affect the simplification
$(G\ominus W)^-$ of a given DAG or network $G$. By the results developed in~\cite{HL:24}, cf.
Observation~\ref{obs:prune-basics}, the DAG $H\coloneqq(G\ominus W)^-$ always preserves the
$\preceq_G$-relation i.e. $u\preceq_{H}v$ if and only if $u\preceq_G v$ for all $u,v\in V(H)$.
Consequently, $\mathfrak{C}(H)\subseteq\mathfrak{C}(G)$, provided that $W$ contains no leaves of
$G$. Thus, the simplification $H$ never introduces clusters that were not already present in $G$.

In another more general direction, in \cite[p.4]{Francis:25}, Francis states
\begin{quotation}
    ``The normalization captures the key information in \emph{any} given network: the information
    that is in fact reconstructable, that is, the vertices and edges whose evolutionary signal may
    be visible in the leaves. If one has data that evolved from a process on a network,
    reconstructing the normalization of the network is the sweet spot: both the most information
    reasonable to hope for, and the most that can practically be obtained.''
\end{quotation}
This provides a compelling argument for the normalization procedure in particular, and for normal
networks more generally. Although normalization is designed to remove vertices that are not visible,
visibility is a path-based notion of relevance that depends on the topology of the entire network.
From the perspective of data-driven simplification, it may therefore not be always clear which
observable data would directly support such a topological property, especially when the underlying
network is unknown. For example, depending on the data at hand, it is conceivable that one may be
able to infer the existence of a vertex $v$ with a particular cluster $C=\CC_N(v)$, even if no
vertex $u$ with $\CC_N(u)=C$ is visible in the network $N$. Such a cluster would disappear under
normalization, since no vertex with cluster $C$ would be retained; cf.\
Proposition~\ref{prop:normN-basics}(3).

In general, such considerations could be very important when trying to take into account the data
when simplifying networks. For example, there is evidence that relationships between LCAs can, to
some extent, be inferred from so-called best hits~\cite{Geiss2020,Stadler2020}. More precisely,
pairwise similarities between genomic sequences associated with leaves $a$, $b$, and $b'$, where $a$
belongs to some species $A$ and $b,b'$ belong to a species $B$ distinct from $A$, may support
conclusions of the form $\lca_N(a,b)\preceq_N\lca_N(a,b')$ even without prior knowledge of the
structure of the underlying network $N$. The vertices $\lca_N(a,b)$ and $\lca_N(a,b')$ need not be
visible, however, and may therefore be removed by the normalization $\norm(N)$. Thus, it may be
useful to consider more general simplifications of the form $(N\ominus W)^-$, where $W$ need not
contain all non-visible vertices. Such a choice of $W$ allows one to retain non-visible vertices
whose relevance is nevertheless supported by the available data, for instance because they induce
observed clusters or occur as unique LCAs of pairs, or more generally subsets, of leaves.

In summary, the main point is that, depending on the type of data under consideration, there may be
different notions of what constitutes the \emph{key information} of a network and, consequently, of
what makes a vertex relevant for simplification. Thus, network simplifications should aim to retain
those vertices whose existence is supported, or witnessed, by the available data. If the relevant
signal is visibility, then this requires information about paths from vertices to leaves. If the
data are clusters, then vertices $v$ with $\CC_N(v)$ among the supported clusters should be
retained, since their relevance is witnessed by the cluster system. Similarly, if the data support
the existence or relative position of LCAs of pairs of leaves, then the corresponding LCA vertices
should be retained. In this sense, the appropriate simplification of a network should be guided by
the kind of evolutionary signal present in the data, rather than by a single universal notion of
vertex relevance and, indeed, this information could also be useful for the more challenging problem
of constructing a network in the first place.

\begin{appendix}

\section{Outsourced Proofs}
\label{sec:appx}

  \begin{proof}[Proof of Lemma~\ref{ominus:commutative}]
      Let $G$ be a DAG and $u,v\in V(G)$ be distinct vertices. By definition, the vertex sets of $(G \ominus v) \ominus u$ and 
    $(G \ominus u) \ominus v$ coincide. We show now that their arc sets coincide. To this end, let 
     $p\to q$ be an arc of $(G \ominus v) \ominus u$
     and thus, $\{p,q\}\cap \{u,v\}=\emptyset$
     and, in particular, $|\{p,q, u,v\}|=4$. We will show that $p\to q$ is an arc of $(G\ominus u)\ominus v$.

     If $(p,q)\in E(G)$ then it is, by definition, also an arc in $(G \ominus u) \ominus v$. Suppose now that 
    $(p,q)\notin E(G)$. Now, one of the following cases occurs. 
   \begin{owndesc} 
            \item[Case: $(p,q)\in E(G\ominus v)$.] Since
            $(p,q)\notin E(G)$, it holds that $p\in\parent_{G}(v)$ and $q\in\child_{G}(v)$. Since $p,q,v\neq u$,
            the arcs $p\to v$ and $v\to q$ remain in $G\ominus u$ i.e. $p\in\parent_{G\ominus u}(v)$ and $q\in\child_{G\ominus u}(v)$.
            Thus, by definition, $(p,q)\in E((G \ominus u) \ominus v)$. 
            \item[Case: $(p,q)\notin E(G\ominus v)$.] In this case, $(p,q)\in E((G \ominus v) \ominus u)\setminus E(G\ominus v)$ hold. Hence, by definition,
             $p\in\parent_{G\ominus v}(u)$ and $q\in\child_{G\ominus v}(u)$. In other words, $(p,u),(u,q)\in E(G\ominus v)$. Now one of the following subcases occurs.
            \begin{owndesc}
                \item[Subcase: $(p,u),(u,q)\in E(G)$.] 
                By definition, $(p,q) \in E(G\ominus u)$. 
                Since, $p,q\neq v$ it follows that 
                $(p,q) \in E((G\ominus u)\ominus v)$.
                \item[Subcase: $(p,u)\notin E(G),(u,q)\in E(G)$.]  Since $(p, u)$ is an arc of $G\ominus v$ but not of $G$, the arcs $p\to v$ and $v\to u$ must exist in $G$.                
                Since also $u\to q$ is an arc of $G$, it follows that there is a path $p\to v\to u\to q$ in $G$. By definition, $p\to v$ and $v\to q$ are arcs of $G\ominus u$ and it follows that $p\to q$ is an arc of $(G\ominus u)\ominus v$.
                \item[Subcase: $(p,u)\in E(G),(u,q)\notin E(G)$.] Since $u\to q$ is an arc of $G\ominus v$ but not of $G$, $u\to v$ and $v\to q$ must be arcs of $G$. Thus, there is a path $p\to u\to v\to q$ in $G$. By similar arguments as in the previous case, $p\to q$ is an arc of $(G\ominus u)\ominus v$. 
                \item[Subcase: $(p,u),(u,q)\notin E(G)$.] As argued
                before, since $(p, u)\in E(G\ominus v)$ 
                and $(p, u)\notin E(G)$, the arc $v\to u$ must exist in $G$.
                Moreover, since $(u, q) \in E(G\ominus v)$ 
                and $(u, q) \notin E(G)$, 
                $u\to v$ is an arc  of $G$.
                The existence of the two arcs $u\to v$
                and $v\to u$ in $G$, however, contradicts the fact
                that $G$ is a DAG. Hence, this case cannot occur. 
            \end{owndesc}
        \end{owndesc}
    
    In summary, in all cases that can occur,
    $(p,q)\in E((G\ominus v)\ominus u)$ implies $(p,q)\in E((G\ominus u)\ominus v)$. By
    analogous argumentation, $E((G\ominus u)\ominus v)\subseteq 
     E((G\ominus v)\ominus u)$ holds and thus,
    the arc sets of the graphs coincide. Therefore,  $(G \ominus v) \ominus u = (G \ominus u) \ominus v$.
  \end{proof}

\begin{proof}[Proof of Lemma~\ref{lem:stronglynormal=>2-lca-rel}]
Let $G$ be a strongly normal DAG on $X$. If $G$ has $k\geq2$ roots $r_1,\ldots,r_k$, we construct the network $N$ by adding a new root $\rho$ and the arc $\rho\to r_i$ for each $1\leq i\leq k$. If $G$, instead, is already a network, we put $N\coloneqq G$ and let $\rho$ denote its unique root. It is easy to see that in both cases, $N =(V,E)$ is a strongly normal network on $X$.
We will now use \cite[Thm~3.9]{Willson2010} established by Willson to prove the result. 
To this end, we define the \emph{base-set} $\widetilde X_N$ of $N$ as the set of vertices in 
$N$ comprising the root, the leaves, and all vertices with out-degree 1. 
In our setting $\widetilde X_N = X\cup \{\rho\}$, since strongly normal networks by definition have no vertices with out-degree 1.
Willson has shown that each vertex $v\in V\setminus \widetilde X_N$
satisfies $v=\lca_N(x,y)$ for some $x,y\in \widetilde X_N$. 
Note that $\lca_N(\{x\})=x$ for all $x\in X$. 
By the latter arguments, all vertices $v\in V \setminus \{\rho\}$ satisfy 
$v=\lca_N(x,y)$ for some $x,y\in X\cup \{\rho\}$. As $v\prec_N \rho$
for all  $v\in V \setminus \{\rho\}$, we obtain $v=\lca_N(x,y)$ for some $x,y\in X$.

Hence, if $G$ has multiple roots, then the latter result 
together with  the fact that $u\preceq_G v$ if and only if $u\preceq_N v$ for all $u,v\in V(G)=V\setminus \{\rho\}$ implies that, for all $v\in V(G)$, there are some $x,y\in X$ such that $v=\lca_N(x,y)$. 
Therefore, $G$ is 2-lca-relevant in case $G\neq N$.

We now consider the case when $N=G$. Hence,  in order to show that 
$N$ is 2-lca-relevant, it remains to show that $\rho=\lca_N(x,y)$ for some $x,y\in X$. 
To this end, we extend $N$ by adding a vertex $r$ and the arc $r\to \rho$, resulting in 
the network $N^*$. It is easy to verify that $N^*$ remains a normal network 
on $X$  with base-set $\widetilde X_{N^*} = X\cup \{r\}$. We can now apply 
 \cite[Thm~3.9]{Willson2010} again to conclude that $\rho=\lca_{N^*}(x,y)$
 for some $x,y\in X\cup \{r\}$. By definition of least common ancestors we have $x,y\preceq_{N^*}\lca_{N^*}(x,y)=\rho\prec_{N^*} r$, which ensures that $r\neq x$ and $r\neq y$.
 In other words, $\rho=\lca_{N^*}(x,y)$ implies $x,y\in X$. 
 It is now straight-forward to check that $\rho=\lca_{N}(x,y)$ holds.  In summary, $G$ is 2-lca-relevant, in the case that $G=N$.
 
 Thus, $G$ is 2-lca-relevant in both cases. In particular, $G$ is $k$-lca-relevant for all $k\in\{2,\dots,|X|\}$.
 \end{proof}

\begin{proof}[Proof of Proposition~\ref{prop:lca-i/reg-i_of_normal}]
    Let $G=(V,E)$ be a normal DAG on $X$ and put $U \coloneqq \{v\in V\mid \outdeg_G(v)=1\}$. Let
    $v \in U$ and denote with $c$ the unique child of $v$. By Lemma~\ref{lem:ominus-basics}, $G\ominus v$ is a DAG on $X$. We start with showing that $G\ominus v$
    remains  normal.
    Since $G$ is tree-child and since $c$ is the unique child of $v$ it must be a
    tree-vertex, i.e., $\indeg_G (c) = 1$, which implies that $v$ is the unique parent of $c$.
        
    We continue with showing that $G\ominus v$ remains tree-child. Let $u\in V(G\ominus v)\setminus
    X$. We consider now the following cases for $u$: (i) $u=c$ and (ii) $u$ is a parent of $v$ and
    (iii) $u\neq c$ and $u$ is not a parent of $v$. 
    
    Case (i): $u=c$. Since $G$ is
    tree-child and $u\notin X$, $u=c$ has a tree-child $c'$ in $G$. Since the children of $c$ in
    $G\ominus v$ are are precisely the children of $c$ in $G$ it follows that $c'$ is a child of $c$
    in $G\ominus v$. Moreover, as $v$ is the unique parent of $c$ and $u\neq v$, it follows that
    $c'$ is not a child of (and since $G$ is a DAG also no parent of) $v$. Hence, the out-degree and
    in-degree of $c'$ in $G$ is the same as in $G\ominus v$, and we can conclude that $c'$ remains a
    tree-child of $u=c$ in $G\ominus v$. 
    
    Case (ii): $u$ is a parent of
    $v$ in $G$. If $u$ is the unique parent of $v$ in $G$, then $u$ is the unique parent of $c$ in
    $G\ominus v$, in which case $c$ becomes the new tree-child of $u$ in $G\ominus v$. If $u$ is one
    of at least two parents, then $v$ is not a tree-child of $u$ in $G$ since $\indeg_G (v) > 1$.
    However, $u$ must have a tree-child $c'$ in $G$. As $v$ is the unique parent of $c$ it follows
    that $c\neq c'$ and, therefore, $c'$ cannot be a child of $v$ in $G$. Moreover, $c'$ cannot be
    a parent of $v$ in $G$ as otherwise the arc $u\to v$ would be a shortcut. Hence, the out-degree
    and in-degree of $c'$ in $G$ is the same as in $G\ominus v$, and we can conclude that $c'$
    remains a tree-child of $u$ in $G\ominus v$. 
    
    Case (iii): $u$
    is neither a parent nor a child of $v$. Since $G$ is tree-child, $u$ must have a tree-child $c'$
    in $G$. Since $u$ is not a parent of $v$, it follows that $c'\neq v$ and, therefore, $c'\in
    V(G\ominus v)$. Since $u\neq v$ and $v$ is the unique parent of $c$, we can conclude that $c\neq
    c'$. It immediately follows, by construction of $G\ominus v$ that the in-degree of $c'$ in $G$
    and $G\ominus v$ remain the same, i.e., $c'$ remains a tree-child of $u$. In summary, all
    vertices $u\in V(G\ominus v)\setminus X$ have a tree-child and thus, $G\ominus v$ is tree-child.

    We continue by showing that $G\ominus v$ remains shortcut-free. Assume, for contradiction,
    that $G\ominus v$ contains a shortcut $(a,b)$. Hence, there is an alternative path $P =
    a\leadsto b$ in $G\ominus v$ that does not contain the arc $(a,b)$. Note that
    $b\prec_{G\ominus v} a$, $v\notin X$ and Observation~\ref{obs:prune-basics} implies that $b\prec_G a$ holds. Suppose first
    that $(a,b)\notin E(G)$. The latter is, by construction of $G\ominus v$, only possible if $c=a$
    or $c=b$ holds for the unique child $c$ of $v$. However, since the children of $c$ 
    and their descendants in $G$ and
    $G\ominus v$ are identical, $a=c$ is not possible as then $(a,b)\in E(G)$. Thus, $c=b$ must hold,
    in which case $a$ must be a parent of $v$ in $G$, i.e., there is an arc $a\to v$ in $G$. Since
    $c$ is the unique child of $v$ in $G$ and $v$ the unique parent of $c$ in $G$, it follows that
    any path from $a$ to $b=c$ must contain the arc $v\to b$. However, there is the alternative
    $P$ from $a$ to $b$ in $G\ominus v$. 
    One easily verifies that this is only possible if there is
    a path $P' = a\leadsto v\to b$ in $G$ where the subpath $a\leadsto v$ contains more than the
    single arc $a\to v$. But then $G$ is not shortcut-free; a contradiction. Therefore, $(a,b)\in
    E(G)$ must hold. However, since $P = a\leadsto b$ is a directed path from $a$ to $b$ 
    in $G\ominus v$ it follows that $v$ is not contained in $P$ and
    Observation~\ref{obs:prune-basics} implies that all vertices $w$ in $P$ satisfy $b\preceq_G
    w\preceq_G a$. As at least one such vertex $w\neq a,b$ exists it follows that there is an
    alternative path $a\leadsto b$ in $G$, a contradiction to $G$ being shortcut-free. Thus,
    $G\ominus v$ is shortcut-free.

    In summary, we have shown that $G\ominus v$ is a shortcut-free and tree-child DAG on $X$, i.e., $G\ominus v$ is
    a normal DAG on $X$. Since, by Lemma~\ref{lem:ominus-outdeg1}, the out-degrees of vertices in $G\ominus v$ are the same in $G$ and $G\ominus v$, the previous arguments and induction can be used to show that
    for all $W\subseteq U$ it holds that $G\ominus W$ is a normal DAG on $X$. 
    As $U$ is the set of all vertices of out-degree 1 in $G$, Lemma~\ref{lem:ominus-outdeg1} also
    implies that $G\ominus U$ cannot have vertices of out-degree 1. 
    Hence,  $G\ominus U$ is, in particular, strongly normal.  

    Finally, we will show that $\notlcaV{i}(G)=U$ for all $i\geq2$. We first show
    that $\notlcaV{2}(G)=U$ by verifying that $V\setminus \notlcaV{2}(G) = V\setminus U$. Since, by
    Lemma~\ref{lem:lca-sets-inclusions}, $U\subseteq \notlcaV{2}(G)$ it follows that $V\setminus
    \notlcaV{2}(G) \subseteq V\setminus U$. Now let $v\in V\setminus U$. By definition, $v\in
    V(G\ominus U)$; see also Lemma~\ref{lem:ominus-basics}. As shown in the previous
    paragraph, $G\ominus U$ is strongly normal and so Lemma~\ref{lem:stronglynormal=>2-lca-rel}
    implies that $G\ominus U$ is 2-lca-relevant. Hence, there are leaves $x,y\in X$ such that $v =
    \lca_{G\ominus U}(x,y)$. In other words, $\LCA_{G\ominus U}(\{x,y\})=\{v\}$. This together with
    Lemma~\ref{lem:ominus-outdeg1} implies that $\LCA_{G}(\{x,y\})=\{v\}$ and, therefore, $v =
    \lca_G(x,y)$. By definition, $v\in \lcaV{2}(G)=V\setminus \notlcaV{2}(G)$. In conclusion, $U =
    \notlcaV{2}(G) $ holds. This together with Lemma~\ref{lem:lca-sets-inclusions} implies that
    $\notlcaV{i}(G)=U$ for all integers $i\geq2$.
\end{proof}

\begin{proof}[Proof of Lemma~\ref{lem:cov(N)-ominus-form}]
    Let $N$ be a network on $X$ with unique root $\rho$. 
    Let $G$ denote the digraph as specified in Definition~\ref{def:normalisation}(1).
    Hence, $\cov(N)=G^-$. Since $V(G)=\vis(N)$ and  $u\to v$ is an arc of $G$ if and only if $v\prec_N u$, it follows
    that there is a directed $uv$-path in $G$ if and only if $u,v\in \vis(N)$ and $v\preceq_N u$.
    Consequently, $G$ is a DAG such that
    \begin{equation}\label{eq:prec-equiv-0}  
        v\preceq_{G} u \iff v\preceq_N u,\quad\text{ for all $u,v\in \vis(N)$.}
    \end{equation} 
    Note that, by definition of visibility, $\rho\in \vis(N)$ and $X\subseteq \vis(N)$. 
    Moreover, for all $v$ in $N$,  $v\preceq_N \rho$
    and there is a leaf $x\in X$ with $x\preceq_N v$. The latter two arguments 
    together with  Eq.~\eqref{eq:prec-equiv-0} ensure that 
    $\rho$ remains the unique root of $G$ and that the set of leaves of $G$ is $X$
    and, in particular, that $G$ is network on $X$
    
    Now, by definition, $\cov(N)=G^-$ and $\cov(N)$ is, by Lemma~\ref{lem:properties-SF-G} and the
    arguments in the previous paragraph, a shortcut-free network on $X$ which satisfy
    $V(\cov(N))=V(G)=\vis(N)$. Moreover, Lemma~\ref{lem:properties-SF-G} together with
    Eq.~\eqref{eq:prec-equiv-0} also ensures that
    \begin{equation}\label{eq:prec-equiv-1}  
        v\preceq_{\cov(N)} u \iff v\preceq_N u,\quad\text{ for all $u,v\in \vis(N)$.}
    \end{equation}
    
    Now consider $N'\coloneqq(N\ominus \notvis(N))^-$. Since $\notvis(N)\subseteq V(N)\setminus X$, 
    we can apply Lemma~\ref{lem:properties-SF-G} and Lemma~\ref{lem:ominus-basics}
    to conclude that $V(N')=V(N)\setminus \notvis(N)=\vis(N)$ and that
    \begin{equation}\label{eq:prec-equiv-2}  
        v\preceq_N u \iff v\preceq_{N'} u,\quad\text{ for all $u,v\in \vis(N)$.}
    \end{equation}
    By, Lemma~\ref{lem:properties-SF-G}, $N'$ is shortcut-free. We have thus shown that $N'$ and
    $\cov(N)$ are shortcut-free DAGs with common vertex set $\vis(N)$. Moreover,
    Eqs.~\eqref{eq:prec-equiv-1} and \eqref{eq:prec-equiv-2} together shows that $u\preceq_{\cov(N)}
    v$ if and only if $u\preceq_{N'} v$ for all $u,v\in \vis(N)$. Thus
    Lemma~\ref{lem:same-prec-so-G=H} implies that $\cov(N)=N'=(N\ominus \notvis(N))^-$.

    Finally, we show that $\cov(N)$ is normal network on $X$.  
    Recall that
    $\cov(N)$ is a network on $X$.  We continue by showing that every vertex of $\cov(N)$ is visible. 
     Assume, for contradiction, that there is a vertex $v$ in $\cov(N)$ that is not visible. 
    Let $x\in X$ be a leaf with $x\preceq_N v$ and $P_1 = \rho\leadsto v\leadsto x$
    be directed path in $\cov(N)$. Since,  $v$ is not visible 
    in $\cov(N)$ there is an alternative path $P_2 = \rho \leadsto x$
    that does not contain $v$. 
    By Eq.~\eqref{eq:prec-equiv-1}, there are paths $P'_1 = \rho\leadsto v\leadsto x$
    and $P'_2 = \rho \leadsto x$ in $N$ where $P'_2$ does not contain $v$. 
    Hence, $v$ is not visible in $N$; a contradiction to $v\in V(\cov(N))=\vis(N)$. 
    Thus, all vertices of $\cov(N)$ are visible. Since $\cov(N) = G^-$ is a shortcut-free
    network, Theorem~\ref{thm:char-normal-vis} implies that $\cov(N)$ is normal.
\end{proof}

\begin{proof}[Proof of Proposition~\ref{prop:normN-basics}]
    Let $N$ be a network on $X$ and consider $\norm(N)$. By Theorem~\ref{thm:cov(N)-ominus-form}, $\norm(N)=\cov(N)\ominus \dvis(N)=(N\ominus \notvis(N))^-\ominus \dvis(N)$. By definition, $\cov(N)$ has vertex set $\vis(N)$ and, by
    Lemma~\ref{lem:cov(N)-ominus-form}, $\cov(N)$ is a normal network on $X$.
    Furthermore, by definition, $\dvis(N)\subseteq\{v\in V(\cov(N))\mid\outdeg_{\cov(N)}(v)=1\}$. The latter two facts together with Proposition~\ref{prop:lca-i/reg-i_of_normal} ensure that $\cov(N)\ominus \dvis(N)$ is a normal DAG on $X$. 
    Since, by Theorem~\ref{thm:cov(N)-ominus-form}, $\norm(N)=\cov(N)\ominus \dvis(N)$, it follows that $\norm(N)$ is a normal DAG on $X$.
    In particular, Lemma~\ref{lem:ominus-basics} ensures that $V(\norm(N))=V(\cov(N))\setminus \dvis(N)=\vis(N)\setminus \dvis(N)$. 
    In addition, Theorem~\ref{thm:cov(N)-ominus-form} implies that  $V(\norm(N))=V(N)\setminus (\notvis(N)\cup \dvis(N))$, and Lemma~\ref{lem:ominus-basics} and Eq.~\eqref{eq:prec-equiv-1} together ensure that 
    \begin{equation*}
        v\preceq_{\norm(N)} u \iff v\preceq_{\cov(N)} u \iff v\preceq_{N} u\quad\text{ for all $u,v\in V(\norm(N))\subseteq V(\cov(N))$.}
    \end{equation*}
    In particular, the unique root $\rho$ of $\cov(N)$ will, by definition, remain a vertex in $\norm(N)$ and satisfies $v\preceq_{\norm(N)}\rho$ for all $v\in V(\norm(N))$. Therefore, $\norm(N)$ has a unique root and is indeed a normal \emph{network} on $X$. We have thus 
    proven statement (1) and (2) in the proposition. Statement (3) is an immediate consequence of (2), since for all $x\in X$ and $v\in V(\norm(N))$ we have $x\preceq_N v$ if and only if $x\preceq_{\norm(N)} v$.

    For the first part of Statement (4), recall that, by definition, $\norm(N)=\cov(N)\ominus \dvis(N)$. It is easily verified that if $v\in\dvis(N)$ i.e. if $\indeg_{\cov(N)}(v)=\outdeg_{\cov(N)}(v)=1$, then $\outdeg_{\cov(N)\ominus v}(u)=\outdeg_{\cov(N)}(u)$ and $\indeg_{\cov(N)\ominus v}(u)=\indeg_{\cov(N)}(u)$ for all $u\in V(\cov(N)\ominus v)$. Consequently, $\norm(N)=\cov(N)\ominus \dvis(N)$ contains no vertex of in- and out-degree 1. The second part of Statement (4) is an immediate consequence of the first part and the definition of phylogenetic networks.

    Next consider Statement (5). The \emph{only-if} part is an immediate consequence of Statement (1) and (4). For the \emph{if} part, suppose that $N$ is normal and that there exists no $v\in V(N)$ such that $\indeg_N(v)=\outdeg_N(v)=1$. By Theorem~\ref{thm:char-normal-vis}, $V(N)=\vis(N)$, so Lemma~\ref{lem:cov(N)-ominus-form} implies that $\cov(N)=(N\ominus \notvis(N))^-=(N\ominus \emptyset)^-=N$, in particular applying that $N$ is shortcut-free.
    Since $\cov(N)=N$ and there is no $v\in V(N)$ such that $\indeg_N(v)=\outdeg_N(v)=1$, we also have $\dvis(N)=\emptyset$. Hence, $\norm(N)=\cov(N)\ominus \dvis(N)=\cov(N)=N$ follows.
    
    The final two statements can be proved verbatim as in the proof of \cite[Thm.~3.3]{francis2021normalising} (see Appendix A therein) and are therefore left to the reader.
\end{proof}

\begin{proof}[Proof of Corollary~\ref{cor:sets}]
 Let $N=(V,E)$ be a network. The equivalence between Condition (1) and (2) has been shown in Theorem~\ref{thm:norm(N)-ominus-form}.
 Suppose Condition (2) holds. Then, $\vis(N) = (\vis(N)\setminus \dvis(N))\cup \dvis(N)=\lcaV{i}(N)\cup \dvis(N)$. Since, by Lemma~\ref{lem:parents/children-in-cov},
 $\dvis(N)\subseteq \notlcaV{i}(N)$ it follows that $\lcaV{i}(N)$ and $\dvis(N)$ are disjoint. Taken the latter two arguments together, $\vis(N)=\lcaV{i}(N)\cupdot \dvis(N)$
 and, thus, Condition (3) holds. Suppose Condition (3) holds.
 Hence, $V\setminus \vis(N)=V \setminus (\lcaV{i}(N)\cupdot \dvis(N))$. 
 Note that $V\setminus \vis(N) = \notvis(N)$ and that 
 $V \setminus (\lcaV{i}(N)\cupdot \dvis(N)) = (V \setminus \lcaV{i}(N))\setminus \dvis(N) = \notlcaV{i}(N) \setminus \dvis(N)$. 
 Hence, $\notvis(N) = \notlcaV{i}(N) \setminus \dvis(N)$ and consequently, 
 $\notvis(N) \cupdot \dvis(N) =  \notlcaV{i}(N)$, i.e., Condition (4) holds. 
 If Condition (4) holds, then $N\ominus (\notvis(N)\cupdot \dvis(N)) = N\ominus \notlcaV{i}(N)$. Thus,
 $(N\ominus (\notvis(N)\cupdot \dvis(N)))^- = (N\ominus \notlcaV{i}(N))^-$. 
 By definition, $\reg_i(N)=(N\ominus \notlcaV{i}(N))^-$ and, 
 by Theorem~\ref{thm:norm(N)-ominus-form}, $\norm(N) = (N\ominus (\notvis(N)\cupdot \dvis(N)))^-$. Therefore, $\norm(N)=\reg_i(N)$ and Condition (1) holds. 
 In summary, Condition (1)--(4) are equivalent. 

 For the final equivalence, first assume that Condition (3), or equivalently, (4) holds. 
 By Condition (3), $\vis(N)=\lcaV{i}(N)\cupdot \dvis(N)$.
 Hence, $\notvis(N)=V\setminus \vis(N) = V\setminus (\lcaV{i}(N)\cupdot \dvis(N))
 \subseteq V\setminus \lcaV{i}(N)$. Therefore, 
 $\lcaV{i}(N)\cap \notvis(N)=\emptyset$. Moreover, by Condition (4), 
 $\notlcaV{i}(N) \cap \vis(N) = (\notvis(N)\cupdot \dvis(N)) \cap \vis(N) =
 (\notvis(N)\cap \vis(N))\cup (\dvis(N) \cap \vis(N)) = \dvis(N)$
 where the latter equation follows from the fact that, by definition,
 $\notvis(N)\cap \vis(N)=\emptyset$ and $\dvis(N)\subseteq \vis(N)$.
 Thus, Condition (5) holds. Assume that Condition (5) holds. 
 By definition, $V(N)=\lcaV{i}(N)\cupdot \notlcaV{i}(N)$ and $V(N)=\vis(N)\cupdot \notvis(N)$. Hence, $\vis(N)=(\lcaV{i}(N)\cupdot \notlcaV{i}(N))\setminus\notvis(N) =(\lcaV{i}(N)\setminus \notvis(N))\cupdot (\notlcaV{i}(N)\setminus \notvis(N))$. Since, by assumption, $\lcaV{i}(N)\cap \notvis(N)=\emptyset$, 
 we obtain $\vis(N)=\lcaV{i}(N)\cupdot (\notlcaV{i}(N)\setminus \notvis(N))$. Moreover, 
 since $V(N)=\vis(N)\cupdot \notvis(N)$, we obtain
  $\notlcaV{i}(N)\setminus \notvis(N)=\notlcaV{i}(N)\cap \vis(N)=\dvis(N)$.
 Hence, $\vis(N) = \lcaV{i}(N)\cupdot \dvis(N)$, i.e., Condition (3) holds.
\end{proof}

\end{appendix}

\subsection*{Data availability}
No data sets are associated with this manuscript.

\section*{Acknowledgements}

We sincerely thank the anonymous reviewers for their careful reading of the manuscript and for
their constructive and detailed comments. Their suggestions helped us to improve the presentation,
clarify the relationship to previous work, and sharpen the exposition of the main contributions.

\bibliographystyle{spbasic}
\bibliography{common}

\end{document}